\documentclass[journal,10pt]{IEEEtran}
\usepackage{graphicx}
\usepackage{lineno}
\modulolinenumbers[5]
\usepackage[cmex10]{amsmath}
\usepackage{amsfonts}

\usepackage{color}
\usepackage{algorithm}
\usepackage{algorithmic}
\usepackage{array}
\usepackage{multirow}
\usepackage{booktabs}
\usepackage{bm}
\usepackage[dvipsnames]{xcolor}

\allowdisplaybreaks[4]
\usepackage{underscore}
\usepackage{stfloats}
\usepackage{float}
\usepackage{lipsum}
\usepackage{cite}
\usepackage{amssymb}
\usepackage{amsmath}
\usepackage{mathrsfs}
\usepackage{amsthm}

\newtheorem{theorem}{Theorem}
\newtheorem{definition}{Definition}

\definecolor{lightblue}{rgb}{0.5, 0.7, 0.90}
\definecolor{lavender}{rgb}{0.6, 0.5, 0.9}
\definecolor{softgreen}{rgb}{0.6, 0.8, 0.6}
\definecolor{deepskyblue}{rgb}{0.0,0.749,1.0}

\usepackage{subfig}               
\usepackage{overpic}   

\makeatletter
\begin{document}

\title{
  Dual-Mind World Models: A General Framework for Learning in Dynamic Wireless Networks 
}

\author{Lingyi Wang,~\IEEEmembership{Graduate Student Member,~IEEE,}
Rashed Shelim,~\IEEEmembership{Member,~IEEE,}\newline
Walid Saad,~\IEEEmembership{Fellow,~IEEE,} and 
Naren Ramakrishnan,~\IEEEmembership{Fellow,~IEEE}

\thanks{A preliminary version of this work was accepted by the IEEE Global Communications
Conference~\cite{wang2025world}, 2025.}
\thanks{This research was supported by the U.S. National Science Foundation under Grant CNS-2225511.}
\thanks{Lingyi Wang, Rashed Shelim and Walid Saad are with the Bradley Department of Electrical and Computer Engineering, Virginia Tech, Alexandria, VA, 22305, USA.
(e-mail: \{lingyiwang, rasheds, walids\}@vt.edu).}
\thanks{
Naren Ramakrishnan is with the Department of Computer Science, Virginia Tech, Alexandria, VA, 22305, USA.
(e-mail: naren@vt.edu).}
}

\maketitle

\begin{abstract}
Despite the popularity of reinforcement learning (RL) in wireless networks, existing approaches that rely on model-free RL (MFRL) and model-based RL (MBRL) are data inefficient and short-sighted. Such RL-based solutions cannot generalize to novel network states since they capture only statistical patterns rather than the underlying physics and logic from wireless data. These limitations become particularly challenging in complex wireless networks with high dynamics and long-term planning requirements. To address these limitations, in this paper, a novel dual-mind world model-based learning framework is proposed with the goal of optimizing completeness-weighted age of information (CAoI) in a challenging mobile, millimeter wave (mmWave) vehicle-to-everything (V2X) scenario. Inspired by cognitive psychology, the proposed dual-mind world model encompasses a pattern-driven System~1 component and a logic-driven System~2 component to learn dynamics and logic of the wireless network, and to provide long-term link scheduling over reliable imagined trajectories. In particular, link scheduling is learned through end-to-end differentiable imagined trajectories with logical consistency over an extended horizon rather than relying on wireless data obtained from environment interactions. Moreover, through imagination rollouts, the proposed world model can jointly reason time-varying network states and plan link scheduling. Thus, during intervals without actual, real-time observations, the dual-mind world model remains capable of making efficient decisions. 
Extensive experiments are conducted on a realistic simulator based on Sionna with end-to-end physical channel, ray-tracing, and scene objects with material properties.
Simulation results show that the proposed world model achieves a significant improvement in data efficiency, and achieves 22\%, 32\%, and 16\% improvement in terms of CAoI, respectively, compared to the state-of-the-art MFRL baseline, MBRL baseline, and the world model approach with only System~1. Moreover, the proposed dual-mind world model achieves strong generalization and adaptation to unseen scenarios and network conditions.
\end{abstract}

\IEEEpeerreviewmaketitle

\vspace{-0.3cm}
\section{Introduction}
Many fundamental wireless networking problems, such as resource management or network control, can be posed as optimization problems~\cite{10636212,shi2023machine, ejaz2025comprehensive,shelim2024fast}. As such, the use of advanced optimization techniques~\cite{10571789,10976568,10363435,10546306,10878807}, ranging from convex optimization to stochastic optimization and dynamic programming, has been instrumental in the evolution of wireless networks towards today's fifth-generation (5G) cellular system and the upcoming sixth-generation (6G) wireless cellular system. However, the limitations of such techniques started to become apparent since 5G. . Particularly, such approaches tend to depend on highly accurate mathematical models of the network, which are difficult to obtain in practice due to stochastic channel variations, user mobility, and incomplete system information~\cite{10891357,10820863,8954939}. Moreover, they cannot satisfy the real-time requirements for complex, non-convex problems. Although heuristic algorithms exist for non-convex problems, such methods often lack scalability and robustness in large, dynamic wireless systems \cite{zhou2023survey}. To alleviate these challenges, there has been a recent surge of works~\cite{chen2024adaptive,10744542,10398469,10943268,wang2023adaptive,9612729} that relied on reinforcement learning (RL) approaches, including both model-based RL (MBRL)~\cite{chen2024adaptive,10744542,10398469} and model-free RL (MFRL)~\cite{wang2023adaptive,10943268,9612729,10077432}, that can learn directly from 
wireless data and adapt to dynamic environments without predefined models.  However, despite the potential advantages of RL compared to traditional optimization approaches, the prior art on RL-based learning approaches~\cite{chen2024adaptive,10744542,10398469,10943268,wang2023adaptive,9612729,10077432} is limited by three significant and fundamental challenges: 
\vspace{-0.2cm}
\begin{enumerate}
  \item \textit{Data inefficiency}: Both MFRL and MBRL face major data inefficiency challenges. For instance, because of their reliance on expensive trial-and-error interactions with the environment, MFRL approaches \cite{wang2023adaptive,10943268,9612729,10077432} cannot efficiently explore a large-scale wireless state-action space and learn an optimal policy in highly dynamic networks without significant environment interactions. Meanwhile, MBRL approaches like those in~\cite{janner2019trust,park2023model,9852968} cannot learn reliable dynamics for wireless networks because the input wireless data, such as high-dimensional channel information, ray-tracing features and interference statistics, is sparse and noisy with uncertainty. Hence, it is difficult to learn an accurate wireless model with limited wireless data.
  \item \textit{Lack of long-term planning abilities}: In a highly dynamic wireless network, there are two main sources for dynamics: (a) uncontrollable exogenous dynamics, such as users' mobility or time-varying channel, and (b) policy-induced endogenous dynamics, such as resource consumption or user state updating. Moreover, in this context, optimality in a single step or within the short term usually cannot ensure global, system objective over long horizons. Hence, there is a need for new optimization and learning approaches that can explicitly model the spatio-temporal causality and logic-driven dependencies of wireless networks, thus supporting physically consistent prediction and long-horizon planning under highly dynamic wireless conditions. However, MFRL methods are inherently short-sighted, as their value estimates rely on immediate rewards and cannot capture long-term dependencies. In contrast, MBRL methods are limited by accumulated model errors over predictions that impair the reliability of long-horizon planning. Moreover, both MFRL and MBRL approaches typically rely on non-differentiable, sampling-based policy learning~\cite{chen2024adaptive,10744542}, thus, they are unable to address the classical credit assignment problem of RL approaches~\cite{moriarty1999evolutionary}, i.e., how to attribute delayed global rewards back to earlier local actions.
  \item \textit{Limited generalization}: Learning a wireless network environment requires machine learning techniques that have strong generalization abilities. Here, generalization refers to the ability of a learning-based approach to transfer the knowledge of underlying wireless physics and dynamics, such as channel variations, blockage patterns, and mobility behaviors, beyond the training data~\cite{shi2023machine}. In other words, a generalizable model can maintain high prediction accuracy and robust control when faced with a stochastic, time-varying wireless environment beyond its original training data. In this context, existing RL approaches mainly rely on statistical pattern recognition of wireless data, but they do not learn the physical propagation characteristics, e.g., blockage, mobility, and channel dynamics, and the causal interaction rules, e.g., scheduling constraints and resource dependencies. Hence, the policies learned by RL approaches often lack robustness and generalization to unseen environments.
\end{enumerate}

\vspace{-0.6cm}
\subsection{Contributions}
\vspace{-0.1cm}
The main contribution of this paper is a novel, universal framework for learning-based wireless network optimization \cite{10929033,chai2025mobiworld,zhao2025world}, grounded in the fundamental framework of \emph{world models}~\cite{hafner2023mastering,sv2023gradient,wang2025dmwm}. In particular, inspired by cognitive psychology, 
we propose a dual-mind world model framework that can capture dynamics and uncertainty of wireless networks, and learn long-term policies in differentiable imagined trajectories with logical consistency over extended horizons. Indeed, this is enabled by the fact that the proposed framework can integrate both fast (so-called System~1) and slow (so-called System~2) thinking abilities~\cite{kahneman2011thinking}. The proposed framework allows a wireless system to learn the underlying physics and logical rules (e.g., logic of link availability and effects of resource scheduling) of its environmental dynamics (e.g., vehicle mobility and frequent link blockages), thus significantly improving data efficiency and providing more reliable imagination over an extended horizon for policy learning. While the proposed framework can apply to a broad range of wireless network problems, we consider a challenging representative scenario pertaining to a millimeter-wave (mmWave) vehicle-to-everything (V2X) communication network and formulate a packet-completeness-aware age of information (CAoI) minimization problem by link scheduling. Particularly, this problem involves both the exogenous dynamics including the vehicles' mobility pattern and real-world channel changes, and the endogenous dynamics of CAoI driven by link scheduling. In summary, our key contributions include:
\vspace{-0.1cm}
\begin{itemize}
  \item We propose a novel world model-based learning framework for wireless networks based on recurrent state-space model (RSSM). Compared to existing RL approaches, RSSM can effectively model the uncertainty and dynamics of the network, significantly enhance data efficiency, i.e., achieve superior task performance within less environment interactions, and endow the wireless network with the long-term planning ability. These improvements are due to the evolution that the policy can be learned in differentiable, end-to-end imagined trajectories from the dynamics model over an extended horizon instead of a short-sighted, expensive trial-and-error mechanism by repetitive environment interactions. 
  \item We further propose a novel dual-mind world model framework tailored to wireless networks,  composed of an intuitive, pattern-driven System~1 component based on RSSM and a logic-driven System~2 component based on logic-integrated neural network (LINN). To overcome the limitations of purely data pattern-driven RSSM, LINN can captures causal and rule-based dependencies in network state transitions, such as how mobility, blockage, and scheduling jointly affect link availability and long-horizon CAoI, thus ensuring logic-consistent imagination of networks' future states and reliable long-term planning. We derive a logic-enhanced evidence lower bound (LE-ELBO) that unifies statistical imagination from System 1 with logical consistency feedback from System 2 to ensure physically consistent predictions.
  \item We develop a realistic simulator based on Sionna and Blender for three-dimensional (3D) dynamic scenario creation and real-world physical channel modeling. The realistic simulator simulates end-to-end channel physics, ray-tracing, and scene objects with material properties. 
  \item Extensive simulation results show that the proposed dual-mind world model achieves a significant improvement in data efficiency, and achieves 22\%, 32\%, and 16\% improvement in terms of CAoI, respectively, compared to the state-of-the-art MFRL, MBRL, and the world model with only System~1. Moreover, the results show that the proposed framework achieves superior generalization and adaptivity to unseen scenarios and network structures.
\end{itemize}
Collectively, these contributions help us create a new framework for wireless network optimization that can more accurately model complex network, integrate fast and slow ``dual-mind'' reasoning to learn data-efficient, long-horizon policies, and generalize robustly across diverse real-world scenarios.


\vspace{-0.4cm}
\section{Related Works} 
Prior works \cite{chen2024adaptive,10744542,10398469,10943268,wang2023adaptive,10077432,9612729} have widely applied RL approaches in mobile wireless networks and age-of-information (AoI) minimization. 
The works in \cite{10398469} and \cite{9612729} considered raw observation information, such as vehicle mobility and channel state information, directly as state input into RL approaches with real-time optimization performance. However, it is challenging for RL approaches to obtain predictive information and learn network physics from the raw data \cite{wang2025dmwm}, and these approaches cannot support long-term planning. 
In \cite{10943268}, the authors addressed a spatial-temporal AoI optimization problem by using a Lyapunov-based decomposition that is coupled with RL. This approach simplifies the optimization but still focuses on short-term decisions, as it cannot capture the long-term dependencies of AoI evolution or estimate future returns from a global perspective.
As previously mentioned, all of the prior works on RL-based wireless network design~\cite{chen2024adaptive,10744542,10398469,10943268,wang2023adaptive,9612729,10077432} are limited in terms of data efficiency, long-term planning, and generalization.
To overcome these challenges faced by RL methods, recent works~\cite{hafner2023mastering,sv2023gradient,wang2025dmwm} in the machine learning community proposed \emph{world model}-based learning frameworks, that could provide a more promising and efficient solution for cognition, prediction and planning. 
Particularly, world models learn and predict the dynamics of the environment along with uncertainty in a latent representation space, which decouples the environment cognition from action planner. Through the imagination ability of the learned environment predictive model,  the planner can be trained by estimating long-term impacts of the current policy in end-to-end differentiable imagined trajectories. In this way, policy learning is independent with actual environment interactions, and future rewards can be accurately attributed to earlier decisions, thereby learning long-term planning abilities~\cite{wang2025world}. World models have been widely used in learning policy from visual data and have shown significant improvement in a broad range of control tasks, ranging from robotic manipulation tasks~\cite{lancaster2024modem} to autonomous navigation and self-driving vehicles~\cite{wang2024driving}. 
However, the existing world models in ~\cite{hafner2023mastering,sv2023gradient,wang2025dmwm}, \cite{lancaster2024modem}, and \cite{wang2024driving} cannot be directly used in wireless networks. For instance, wireless environment observations, such as channel state information and antenna angles, are high-dimensional and sparse, thus it is challenging for existing world model approaches to explore complex spatio-temporal structures from wireless data. Moreover, wireless data is characterized by physical features such as multipath propagation, blockages, and mobility. Hence, a world model pertaining to wireless networks must capture not only the stochastic evolution of wireless states but also the underlying physics and logical structure of communication systems to enable reliable network prediction over extended horizons. Here, we note that, in~\cite{wang2025dmwm}, we proposed cognitive psychology theory-inspired world models for robotic control tasks in environments with smooth and structured dynamics. However, wireless networks exhibit highly spatio-temporally coupled features, where link reliability depends jointly on vehicle mobility, dynamic blockage, and scheduling. Moreover, the wireless network dynamics involve both stochastic in channel variations, and logic in scheduling constraints and physical relationships. These characteristics fundamentally differ from robotic environments and prevent a direct application of the results in~\cite{wang2025dmwm}, thus motivating a specialized world-model design tailored to wireless network optimization.

\begin{figure}
  \centering
  \includegraphics[scale=0.52]{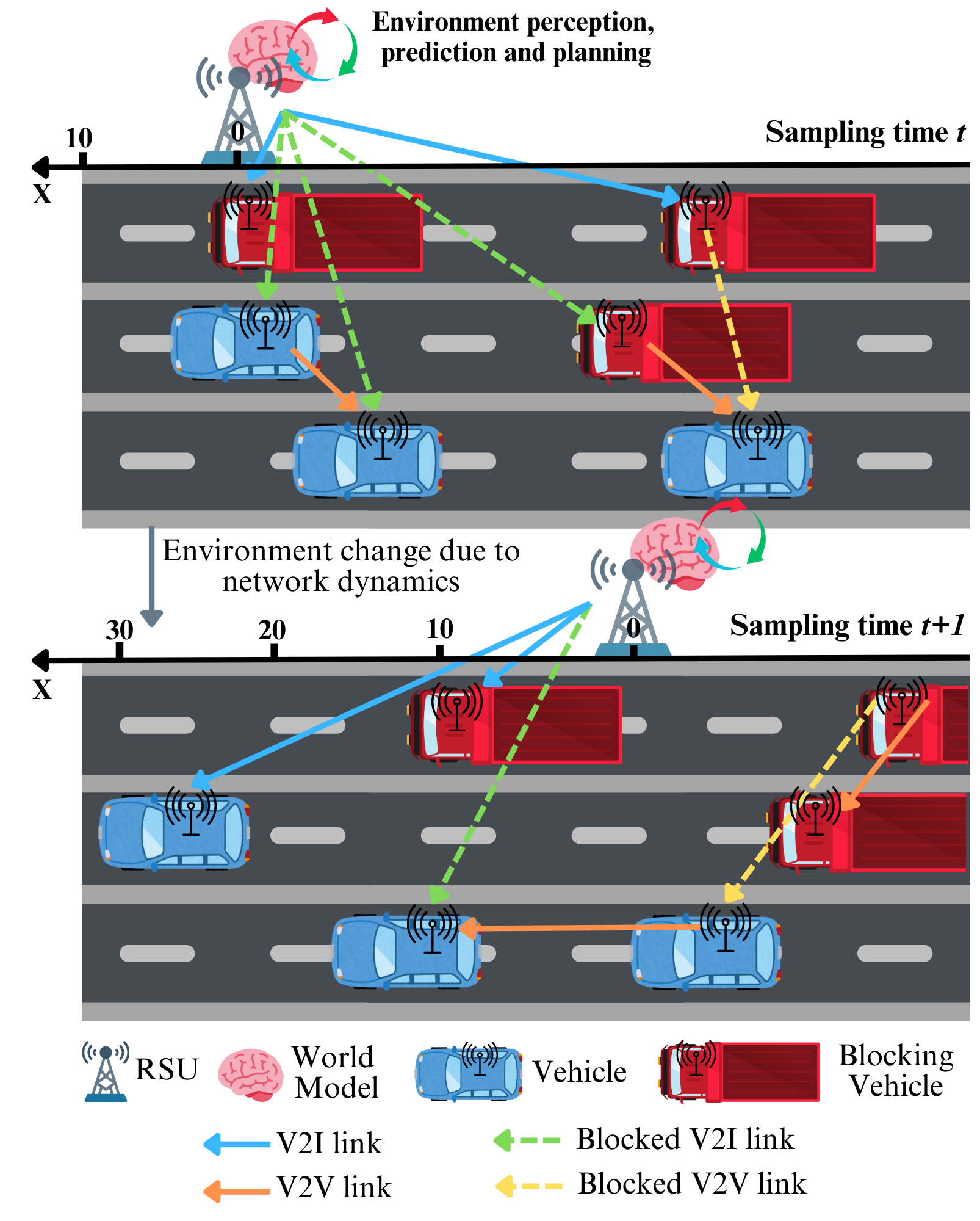}
    \vspace{-0.2cm}
  \caption{Illustration of the use of a world model for learning and optimization in a mmWave V2X communication network.}
  \vspace{-0.6cm}
\end{figure}

\vspace{-0.6cm}
\section{System Model}
\vspace{-0.2cm}
While the world model framework that we will develop in Section IV can apply to a broad range of wireless use cases, to concretely showcase its benefits, we focus on a representative system model, as shown in Fig. 1. We consider a mmWave V2X network consisting of a roadside unit (RSU) $u$ and a set $\mathcal{V}$ of $V$ mobile vehicles, with both vehicle-to-infrastructure (V2I) and vehicle-to-vehicle~(V2V) links. Let $\mathcal{M}_{t}$ and $\mathcal{Z}_{t}$ be, respectively, the learnable, time-varying sets of V2I and V2V link pairs at timeslot $t$. The V2V/V2I links share a bandwidth $B$. Similar to~\cite{tunc2021mitigating}, we use narrow and directional beams, and, thus, there is no interference in the V2X network.
We consider a time-slotted system in which each timeslot is indexed by \( t \) and has a fixed duration $\xi$. Each vehicle operates in a half-duplex communication mode, where it can establish only one communication link during a timeslot $t$, and is unable to transmit and receive data simultaneously.

\vspace{-0.4cm}
\subsection{Transmission Model}
Let $g^{i}_t$ be the mmWave V2X channel gain at timeslot $t$ from the transmitter to the receiver over link $i$, where $g^{i}_t$ is characterized by high path loss, multipath propagation and dynamic blockages.
The data rate, in packets per timeslot $t$, for V2I link $m \in \mathcal{M}_{t}$ and V2V link $z \in \mathcal{Z}_{t}$ is respectively
\begin{equation}
  \begin{aligned}
  &R^{\mathrm{V2V},z}_{t}=\frac{B \xi}{S}  \log_2\left(1+\frac{P_vg^{z}_t}{N_0 B}\right), \\
  &R^{\mathrm{V2I},m}_{t}=\frac{B \xi}{S}  \log_2\left(1+
  \frac{P_u g^{m}_{t}}{N_0 B}\right),
  \end{aligned}
\end{equation}
where $S$ is the size of each packet, $N_0$ is the power spectral density of additive white Gaussian noise, and $P_u$ and $P_v$ are, respectively, the transmit power of the RSU and each vehicle.

In mmWave V2X communication, blockages caused by high-speed mobile vehicles, buildings, and other obstacles significantly impact signal propagation and can lead to disruptions in both V2V and V2I links. To model the blockage effect, we consider the Fresnel zone obstruction~\cite{9838711}, path loss variations, and environmental dynamic characteristics in our blockage model. 
The first Fresnel zone radius determines the critical region for obstruction as
$
\Gamma_\mathrm{F}= \sqrt{\frac{\lambda \delta_{kb} \delta_{bv}}{\delta_{kv}}},
$
where $\delta_{kb}$ and $\delta_{bv}$ are, respectively, the distances from the blocking vehicle to the transmitter and receiver, with $\delta_{kv} = \delta_{kb} + \delta_{bv}$ being the total link distance. A blockage occurs when the height of the blocking vehicle $\Upsilon_b$ exceeds the effective Fresnel height $h_\textrm{F}$, which is given by
$
\Upsilon_\mathrm{F}= \Upsilon_k + \frac{(\Upsilon_v - \Upsilon_k) \delta_{kb}}{\delta_{kv}} - 0.6 \Gamma_\mathrm{F},
$
where $\Upsilon_k$ and $\Upsilon_v$ respectively represent the antennas heights of the transmitter and receiver.
Assume the vehicle heights follow a Gaussian distribution $\Upsilon_b \sim \mathcal{N}(\mu_b, \sigma_b^2)$, the probability of blockage will be
$
\mathrm{Pr}^{\mathrm{block}}_\mathrm{F} = Q\!\left(\tfrac{\Upsilon_\mathrm{F}-\Upsilon_b}{\sigma_b}\right),
$
where $Q(x) = \tfrac{1}{\sqrt{2\pi}} \int_x^{\infty} e^{-t^2/2}\,\mathrm{d}t$ is the Gaussian Q-function. For multiple blocking vehicles, the number of vehicles is assumed to follow a Poisson point process with vehicle density $\lambda_v$ \cite{9838711}, and the line-of-sight (LoS) probability of V2V and V2I links will be, respectively, given by $\mathrm{Pr}_{\text{LoS}}^{{\text{V2V}}} = e^{-\lambda_v \delta_{kv} \mathrm{Pr}^{\mathrm{Block}}_\mathrm{F}}$ and $\mathrm{Pr}_{\text{LoS}}^{{\text{V2I}}} = P_{\text{LoS}}^{\text{3GPP}} \mathrm{Pr}_{\text{LoS}}^{{\text{V2V}}}$,
where $P_{\text{LoS}}^{\text{3GPP}} = e^{-\varepsilon \delta_{kv}}$ is the 3GPP empirical model \cite{giordani2019path} that captures urban blockages from buildings, and $\varepsilon$ is a factor that depends on the environment.

\begin{figure*}[ht!]
  \begin{equation}\label{eq2}
    \begin{aligned}
    &\tilde{G}^{v,m}_t = \mathbb{I}(\lfloor R^{\mathrm{V2I},m}_{t} \rfloor \geq C^u) G^u_{t} + \mathbb{I}(\lfloor R^{\mathrm{V2I},m}_{t}\rfloor < C^u)\left[\frac{\lfloor R^{\mathrm{V2I},m}_{t} \rfloor}{C^u}G^u_{t} + \frac{C^u - \lfloor R^{\mathrm{V2I},m}_{t} \rfloor}{C^u}A^v_t\right],\\
    &\tilde{G}^{v,z}_t = \mathbb{I}(\lfloor R^{\mathrm{V2V},z}_{t} \rfloor \geq C^{v^{\prime}}_{t}) \left[\frac{C^{v^{\prime}}_t}{C^u}G^{v^{\prime}}_t + \frac{C^u - C^{v^{\prime}}_{t}}{C^u}A^v_t\right] \! + \! \mathbb{I}(\lfloor R^{\mathrm{V2V},z}_{t} \rfloor < C^{v^{\prime}}_t)\! \left[\frac{\lfloor R^{\mathrm{V2V},z}_{t} \rfloor}{C^u}G^{v^{\prime}}_t + \frac{C^u - \lfloor R^{\mathrm{V2I},m}_{t} \rfloor}{C^u}A^v_t\right].\\
  \end{aligned}
  \tag{3}
  \end{equation}
  \hrulefill
  \vspace{-0.5cm}
\end{figure*}

\vspace{-0.5cm}
\subsection{CAoI Metric}
The RSU transmits road information that consists of $C^u$ packets in each timeslot. This information includes real-time data such as traffic signal timing, roadside sensor messages, and emergency warnings. We consider a practical broadcast scenario in which the RSU distributes up-to-date data to vehicles for both driving efficiency and safety. This scenario is usually evaluated by AoI to quantify end-to-end latency. However, the classical AoI metric overlooks two key aspects of highly dynamic networks: (a) link reliability, since it does not account for packet loss or partial transmissions caused by blockage and mobility, and (b) temporal scheduling dependence, as AoI considers each update independently and cannot capture how current scheduling influence future information freshness. In particular, when blockages or severe path loss occur in mmWave networks, packets can be truncated and partially received. Hence, in the next definition, we introduce the concept of CAoI by adding packet completeness that scales the AoI by each link's transmission rate to more accurately capture the freshness of successfully delivered information. 

\begin{definition}
The \emph{CAoI} $A^v_{t+1}$ of a vehicle $v\in \mathcal{V}$ at timeslot $t+1$ in the V2X communication network is given by
\begin{equation}
  \begin{aligned}
    A^v_{t+1} = 
\begin{cases} 
  t -  \tilde{G}^{v,m}_t + 1, &v \text{ receives from V2I pair } m,  \\ 
  t - \tilde{G}^{v,z}_t + 1, & v \text{ receives from V2V pair } z,  \\ 
A^v_t + 1,& \text{otherwise}.
\end{cases}
  \end{aligned}
  \tag{2}
\end{equation} 
\end{definition}

\noindent The updated CAoI of vehicle $v$ over V2I link $m$ or V2V link $z$ are represented by $\tilde{G}^{v,m}_{t}$ and $\tilde{G}^{v,z}_t$, respectively, as given by (\ref{eq2}), where $C^{v^{\prime}}$ is the number of expected fresher packets from vehicle $v^{\prime}$, and
$G^{u}_t$ and $G^{v^{\prime}}_t$ are, respectively, the CAoI of the RSU and vehicle $v^{\prime}$.
The indicator function \(\mathbb{I}(x)\) is a binary-valued function that equals to 1 if the condition \( x \) holds true and 0 otherwise.
\addtocounter{equation}{2}

\subsection{CAoI Minimization Problem}
The objective of the network is to minimize its average CAoI by optimization link scheduling over a time period $T$ for packet update, which can be posed as an optimization problem:
\vspace{-0.4cm}
\begin{subequations}\label{P:CAoI}
  \begin{align}
    &\min_{\left\{\mathcal{M}_{t}, \mathcal{Z}_{t}\right\}} \frac{1}{T} \sum_{t=1}^{T} \sum_{v=1}^{V} A^v_t \\ 
   \text{ s.t. }  
   & A^v_t \leq \bar{A}, \forall v \in \mathcal{V}, \\
   & \Phi_{\cap}(\mathcal{M}_{t},\mathcal{Z}_{t}) = \emptyset, 
  \end{align}
\end{subequations}
where $\bar{A}$ is the maximum age tolerance for each vehicle, $\Phi_{\cap}(\mathcal{M}_{t},\mathcal{Z}_{t})$ represents the shared link node (transmitter or receiver) set between the V2I link set $\mathcal{M}_{t}$ and the V2V link set $\mathcal{Z}_{t}$.
It is challenging to optimize the link scheduling in (4) due to the coupled spatial mobility of the V2X network and temporal impacts of link scheduling. Particularly, the network must learn an optimal policy in the presence of both exogenous dynamics, i.e., physical location changes and channel changes, and endogenous dynamics, i.e., CAoI update by policy. In other words, the link scheduling needs to jointly recognize the system's inherent mobility pattern and consider its future influences on the system. 
While many existing works have addressed related optimization problems, such as mobility-aware scheduling~\cite{9612729}, V2V communication under dynamics~\cite{10847910}, and AoI optimization over long horizons~\cite{10943268}, they typically rely on single-step decision-making or short-term policy learning, thus they cannot robustly handle the delayed effects of policies on future information freshness.
In contrast, the CAoI objective considered here couples state transitions and action effects over long timescales. It is also sensitive to long-term feedback loops and coupled dynamics that simple policy models fail to resolve. This motivates the need for a more structured and spatial-temporally expressive solution. In particular, a world model framework can both provide reliable predictions over long horizons with logical consistency and directly learn a long-term policy, which is then realized in the next section.

\begin{figure*}[ht!]
  \centering
  \includegraphics[width=0.75\linewidth]{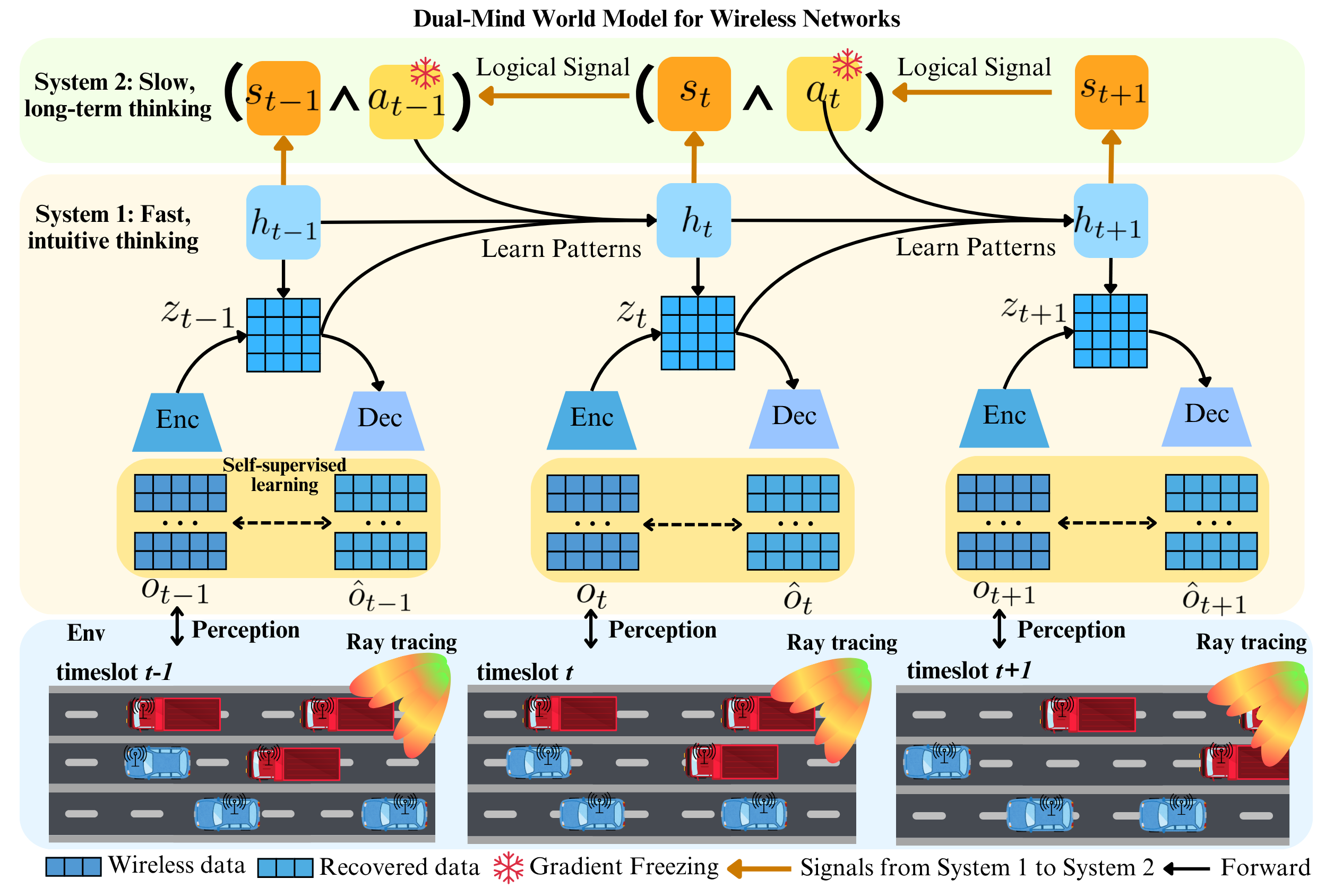}
  \vspace{-0.3cm}
  \caption{Learning a world model for V2X communication networks based on the location data and the ray tracing data.}
  \vspace{-0.6cm}
  \label{Fig:System}
\end{figure*}

\vspace{-0.3cm}
\section{Dual-Mind World Model For Long-Term Prediction and Link Scheduling}
\vspace{-0.1cm}
In this section, we propose a novel dual-mind world model framework for wireless networks, as shown in Fig.~\ref{Fig:System}, which is deployed at the RSU to solve the CAoI minimization problem (\ref{P:CAoI}). Inspired by cognitive psychology, the proposed dual-mind world model consists of a pattern-driven System~1 component for fast inference and a logic-driven System~2 component for capturing logical relationships between network states and actions. 
The advantage of this framework is that it can learn a foresighted planning ability by reliable predictions with long-term logical consistency.
To enable cross-system collaboration, we develop an efficient inter-system signal mechanism. Then, long-term link scheduling is learned through reliable, differentiable imagined trajectories of the wireless network. In the considered scenario, the notion of ``imagined trajectories'' specifically refers to the predictions of CAoI, vehicle locations, and channel states in latent spaces under a given policy. More generally, imagined trajectories can refer to predictive state transitions of the wireless network. Finally, we present a practical use case of joint prediction and link scheduling, that can solve (\ref{P:CAoI}) without real-time wireless data during communication-constrained intervals. Here, we note that this framework will build on and extend our earlier work in \cite{wang2025dmwm}. In particular, the world model developed in \cite{wang2025dmwm} cannot effectively handle the spatio-temporally coupled, hybrid stochastic-logical dynamics of wireless networks, where vehicle mobility, blockage, and link scheduling jointly decide information freshness and link reliability. Hence, we will extend it to a specialized dual-mind world model that integrates network physics and scheduling logic, thus enabling joint learning, prediction and planning for highly dynamic wireless networks.

\vspace{-0.4cm}
\subsection{Pattern-Driven System~1 for Fast Inference}
The System~1 component aims to learn data patterns and environment dynamics from observed wireless data $o_t = \{A_{t},\Xi_{t},L_{t}\}$, that include CAoI $A_{t}=\{A^v_t\}_{v=1}^{V}$, physical channel data $\Xi_{t}$, and vehicle locations $L_{t}$. Particularly, the CAoI of all vehicles enables the network to capture the endogenous dynamics, i.e., dynamic information freshness caused by the current policy, while vehicle locations and physical channel information provide the endogenous dynamics, i.e., spatial geometric relationships among transceivers and physical channel changes caused by the system's inherent pattern. In practice, the location information, CAoI, and channel states can be obtained through periodic status reports from vehicles. Specifically, vehicle positions are available from on-board GPS sensors, while channel data and packet-completeness indicators are fed back to the RSU through control signaling as part of standard V2X protocols. Although such feedback can be unreliable by the occasional loss, delay, or quantization errors in highly mobile environments, these effects can be effectively mitigated by the proposed world model through its joint prediction and planning capability, which enables reliable estimation of missing or outdated information, which will be discussed in Section IV.D..
The planner of the world model decides the link scheduling $a_t=\left\{\mathcal{M}_{t}, \mathcal{Z}_{t}\right\}$ based on the state representations from the System~1 component.

\subsubsection{RSSM-Based Pattern Learning}
For the System~1 component, we use the RSSM framework~\cite{hafner2023mastering}. Particularly, RSSM learns state transitions of the V2X network in a latent space with recurrent structures and variational inference. The RSSM-based System~1 component is used to perform quick, intuitive thinking, and is defined by the following components:
\vspace{-0.1cm}
\begin{align}
    {\raisetag{2.5\normalbaselineskip}
    \begin{array}{ll}
    \text{Deterministic state:}  &h_t = f_\varphi\left(h_{t-1}, z_{t-1}, a_{t-1}\right),\\
    \text{Encoder:}  &z_t \sim q_\varphi\left(z_t \mid h_t, o_t\right), \\
    \text{Stochastic state:}  &\tilde{z}_t \sim p_\varphi\left(\tilde{z}_t \mid h_t\right), \\
    \text{Reward predictor:}  &\tilde{r}_{t} \sim p_\varphi\left(\tilde{r}_{t} \mid h_t, z_t\right), \\
    \text{Decoder:}  &\hat{o}_{t} \sim p_\varphi(\hat{o}_{t} \mid h_t, z_t),
    \end{array}}  
\end{align} 
where $z_t$ is the latent representation for the network observation $o_t$, $o_t$ is the multi-modal representations of $o_t$, $h_t$ is the deterministic state, $\tilde{z}_t$ is the predicted latent representation for the future network state, $\hat{o}_{t}$ is the recovered observations, and $\tilde{r}_{t}$ is the predicted real-world reward at timeslot $t$. For RSSM, we define a loss function 
$\mathcal{L}_{\text{S1}}(\varphi) = \mathcal{L}_{\text {pred }}(\varphi) + \delta_{1} \mathcal{L}_{\text {dyn }}(\varphi) + \delta_{2}  \mathcal{L}_{\text {rep }}(\varphi)$ with the weight factors $\delta_{1}$ and $\delta_{2}$, where:
\vspace{-0.1cm} 
\begin{equation}
  \mathcal{L}_{\text {pred }}(\varphi)  =  -  \ln p_\varphi  \left(\hat{o}_{t} \mid  z_t, h_t\right)-\ln p_\varphi  \left(\tilde{r}_{t}  \mid  z_t, h_t\right),
\end{equation} 
which is a prediction loss that ensures $z_t$ captures features from wireless data $o_t$ and learns the credit assignment $\tilde{r}_{t}$. 
The dynamic loss $\mathcal{L}_{\text {dyn}}$ and the representation loss $\mathcal{L}_{\text {rep}}$ are, respectively, given by
\vspace{-0.2cm}
\begin{equation}\mathcal{L}_{\text {dyn }}(\varphi) =  \mathrm{D}_{\textrm{KL}}\left[\operatorname{sg}\left(q_\varphi\left(z_t  \mid h_t, o_t\right)\right) \| p_\varphi\left(\tilde{z}_t \mid h_t\right)\right],
\end{equation}  
\begin{equation}\mathcal{L}_{\text {rep }}(\varphi) =  \mathrm{D}_{\textrm{KL}}\left[q_\varphi\left(z_t \mid h_t, o_t\right) \| \operatorname{sg}\left(p_\varphi\left(\tilde{z}_t \mid h_t\right)\right)\right].
\end{equation} 
(7) and (8) ensure that $z_t$ and $h_t$ extract the network dynamics in the latent space, where $\operatorname{sg}(\cdot)$ is the stop-gradient operator, and $ \mathrm{D}_{\textrm{KL}}(\cdot)$ is the Kullback-Leibler divergence.

\begin{figure}[!t]
  \centering
  \includegraphics[width=0.8\linewidth]{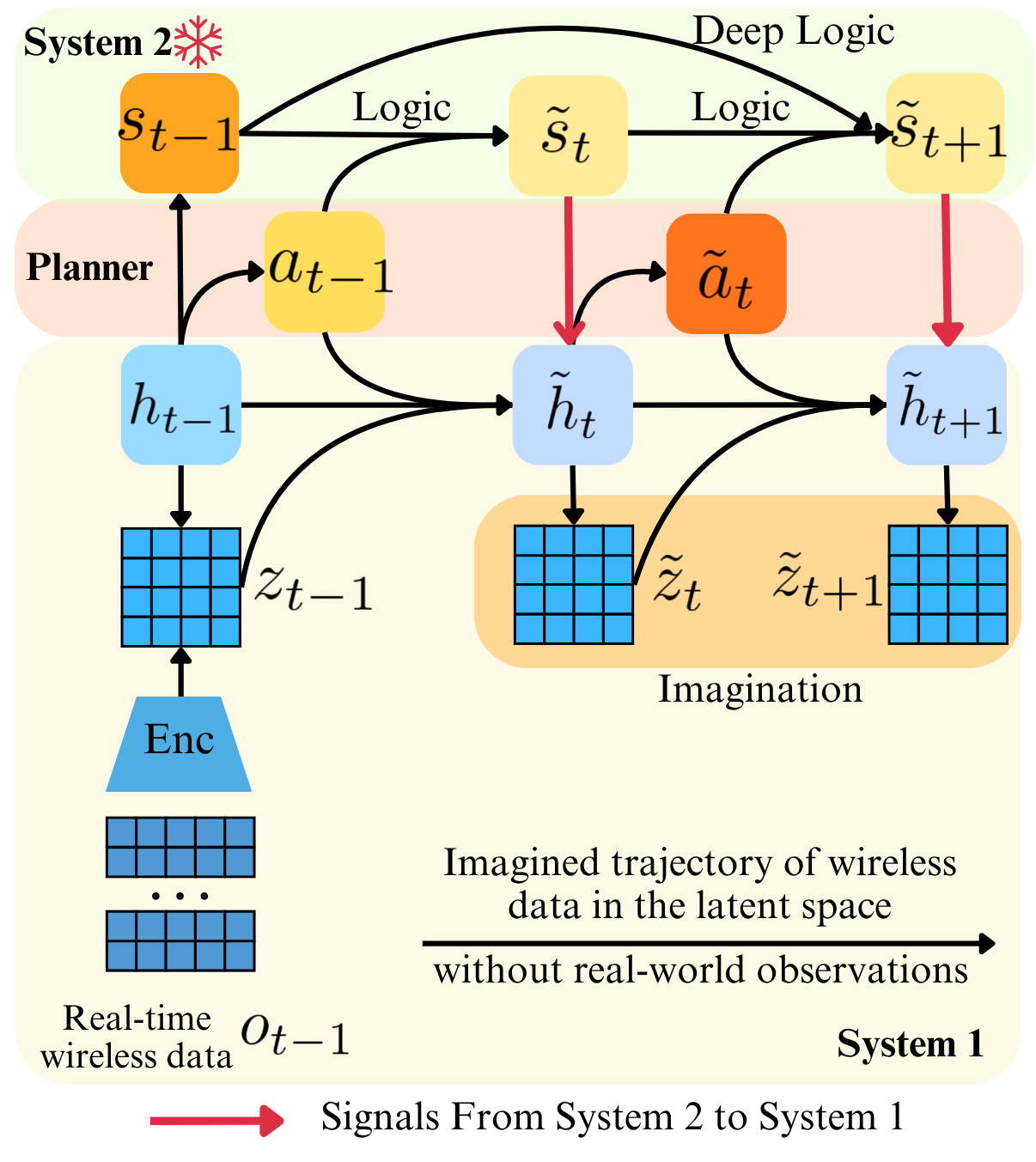}
    \vspace{-0.3cm}
  \caption{The logic-enhanced imagination ability of the proposed dual-mind world model for both policy learning, and joint prediction and scheduling without wireless data.}
  \vspace{-0.5cm}
  \label{Fig:Imag}
\end{figure}

In our proposed dual-mind world model, the RSSM-based System~1 captures the statistical dynamics of the V2X network by learning compact latent representations from observed wireless data, including CAoI, physical channel states, and vehicle locations. While this pattern-driven module enables fast and scalable inference, it is inherently limited in its ability to reason about long-term consequences of actions in the mmWave V2X environment. Specifically, System~1 relies purely on learned correlations from past data and lacks structural understanding of the underlying wireless mechanisms, such as mobility-induced channel disruptions, non-linear CAoI resets, and blockage-driven link changes. Hence, over extended prediction horizons, System~1 may not preserve the logical consistency of network state transitions. This is particularly problematic for highly complex wireless networks that require reliable, long-term planning under uncertainty, delay feedback, and physical dynamics. Thus, to overcome this limitation, we complement RSSM with a System 2 component with a logical reasoning ability, that capture and learn the logic of network state transitions.

\vspace{-0.5cm}
\subsection{Logic-Driven System~2 for Deep Inference}
We now introduce a System~2 component to learn the underlying logical relationships of wireless physics from actual network state transitions. Particularly, we will use the concept of LINN~\cite{shi2020neural} as the foundational module for System~2. Based on LINN, we further propose a deep logical reasoning framework to capture and infer the \textit{logical rules of the wireless network}, i.e., the causal, constraint-based relations among mobility, channel state, link availability, and scheduling decisions that decides how CAoI evolves. In particular, logical operations, including negation (\(\neg\)), disjunction (\(\lor\)), conjunction (\(\land\)) and implication (\(\rightarrow\)), are used to enable symbolic reasoning over discrete and structured relationships that cannot be captured purely by statistical models in wireless networks. For instance, in the considered V2X network, an effective scheduling decision should satisfy logical conditions, such as ``if a link is blocked, it cannot be scheduled,'' or ``if two links share a common node, they cannot be activated simultaneously,'' which reflects physics and logic-driven behavior rather than purely probability and statistics-driven behavior.

\subsubsection{Neural Network-Based Logic Operations} To endow the wireless network with the ability of logical thinking, the world model must capture the structural logical information among observations and actions. Similar to~\cite{shi2020neural}, we use neural networks to realize  
the logic operations of $\mathrm{AND}$, $\mathrm{OR}$, and $\mathrm{NOT}$, which can be respectively represented by 
\vspace{-0.1cm}
\begin{equation}
    \mathrm{AND}(d,a)=\boldsymbol{W}^a_{2}\sigma\left(\boldsymbol{W}^a_{1}(d\oplus a)+\boldsymbol{b}^a_1 \right) +\boldsymbol{b}^a_2,
\end{equation}
\begin{equation}
\mathrm{OR}(d,a)=\boldsymbol{W}^o_{2}\sigma\left(\boldsymbol{W}^o_{1}
    (d\oplus a)+\boldsymbol{b}^o_1 \right)+\boldsymbol{b}^o_2,
\end{equation}
\begin{equation}
\operatorname{NOT}(w)= \boldsymbol{W}^n_2 \sigma\left(\boldsymbol{W}^n_1 w+\boldsymbol{b}^n_1\right)+ \boldsymbol{b}^n_2 , w\in\{a,z\},
\end{equation}
where $\boldsymbol{W}^l_{1}$, $\boldsymbol{W}^l_{2}$, $\boldsymbol{b}^l_{1}$, $\boldsymbol{b}^l_{2}$ are the parameters of a logical neural network, $l = \{a,o,n\}$, and $\sigma(\cdot)$ is the activation function. 
Based on the basic logical operations (9)-(11),  the implication operation $\rightarrow$ is proposed to enable reasoning based on observations and actions for imagined state trajectories of the wireless network. Since the equivalence relationship of $\rightarrow$ is represented by $p \rightarrow q \Longleftrightarrow \neg \, p \vee q$,
we realize the operation IMPLY based on $\neg$ and $\land$, which is formally given by
$\operatorname{IMPLY}(d,a)= \operatorname{OR}(\operatorname{NOT}(z),a)$.
The neural logic operators in System~2 are designed to capture nonlinear, non-geometric logical relationships between wireless observations and scheduling actions that cannot be adequately modeled by standard geometric operations in vector space. For instance, the logical negation of a latent representations $z$, represented by $\mathrm{NOT}(z)$, represents the opposite logical condition in the network. If $z$ encodes a LoS condition, then $\mathrm{NOT}(z)$ represents a non-LoS (NLoS) condition, instead of a simple orthogonal vector $z^{\perp}$ in Euclidean space. 

To ensure that the learned operations (NOT, AND, OR, IMPLY) behave in a logically consistent manner, we incorporate a set of regularization rules derived from classical logic, as shown in Table I. These rules, such as double negation, identity, and complementation, are realized during training by penalizing violations through a regularization loss. In our formulation, the logical constants True and False are represented as fixed vectors $\mathbf{T}$ and $\mathbf{F}$, with $\mathbf{F}=\mathrm{NOT}(\mathbf{T})$. Each logical identity is converted into a differentiable constraint by measuring the similarity between the left and right sides of the rule using a cosine similarity function. These logical constraints regularize the neural operators by enforcing consistency between learned representations and the underlying causal rules of the wireless network. Hence, the System~1 component can learn not only from statistical data patterns but also from the causal, structural relations among mobility, channel state, link availability, and scheduling decisions. Such learning process improves generalization to unseen network states by preventing the model from producing physically or logically inconsistent predictions.
Taking the double negation of the operation $\neg$ as an example, the logical equation $\neg (\neg w) = w$ can be converted into a logical regularization item as $r_1 = \sum_{w \in W} 1 - \text{Sim}(\text{NOT}(\text{NOT}(w)), w)$, where $w\in\{a,z\}$, and $\mathrm{Sim}(w_1,w_2) = \sigma \left((w_1 \cdot w_2)/(\| w_1 \| \| w_2 \|) \right)$ measures the similarity between $w_1$ and $w_2$. Hence, the logical regularization loss can represented by $\mathcal{L}_{\mathrm{reg}} = \sum_{i} r_i$, where $r_i$ is the regularization item of each logical equation in Table \ref{Tab:Logic}.

\subsubsection{Proposed Deep Logical Reasoning}
We now propose a novel deep logical reasoning approach. Particularly, the logical relationships are explicitly captured from the state transitions of the wireless network through the logical operations (9)-(11), and then if-then rules of $s\land a \rightarrow s^{\prime}$ are learned for reasoning chain. First, the logical information $\eta_{t}$ of the premise ($o_t$, $a_t$) at timeslot $t$ can be extracted by
\vspace{-0.2cm}
\begin{equation}
    \eta_{t} \triangleq (o_t \land a_t) = \mathrm{AND}(o_t, a_t), \, \forall t.
\end{equation}
Then, the implication operation $\rightarrow$ is used to align local logic between the premise ($o_t$, $a_t$) and the conclusion ($o_{t+1}$), which can be represented by
$
    \phi_t \triangleq (\eta_{t} \rightarrow o_{t+1}) = \mathrm{IMPLY}(\eta_{t},o_{t+1}), \forall t.
$
Although $\phi_t$ captures single-step logical information for the network state transitions, it cannot capture logical dependence among network states and link scheduling over long horizons for the complex CAoI minimization (\ref{P:CAoI}) that requires the long-term planning. Hence, we propose the deep recursive implication reasoning approach given by:
\vspace{-0.15cm}
\begin{equation}
  \begin{aligned}
    & \phi^{\alpha}_{t} \triangleq (\eta_{t-\alpha} \cdots \land \eta_{t-1} \land \eta_{t} \rightarrow o_{t+1})\\
    & \qquad = \mathrm{IMPLY}(\mathrm{AND}(\cdots,\mathrm{AND}(\eta_{t-1},\eta_{t})),o_{t+1}),
  \end{aligned}
\end{equation}
where $\alpha < t$ represents inference depth. 
The recursive logic $\phi^{\alpha}_{t}$ models the temporal propagation of logical dependencies within wireless networks, and ensures global logical consistency by recursively linking the logical state at time $t$ to those of earlier slots, thereby encoding long-term causal dependencies among network states. The reasoning chain for deep thinking that integrates both local and global logical relationships of the network over a period $T$ is given by
\vspace{-0.2cm}
\begin{equation}
  \begin{aligned}
    & L^{\alpha}_{T}  \triangleq (\phi^{\alpha}_1 \land \phi^{\alpha}_2 \land \phi^{\alpha}_3 \cdots \phi^{\alpha}_{T-1} \rightarrow \mathbf{T}),\\
    & = \mathrm{IMPLY}(\mathrm{AND}(\cdots,\mathrm{AND}(\phi^{\alpha}_{T-2},\phi^{\alpha}_{T-1})),\mathbf{T}).
    \end{aligned}
\end{equation}
The logical loss with inference depth $\alpha$ can be represented by
\vspace{-0.2cm}
\begin{equation}
    \mathcal{L}^{\alpha}_{\text{log}} = \frac{1}{T-1} \sum_t \mathrm{Sim}(\phi^{\alpha}_{t}, \mathbf{T}) - \mathrm{Sim}(\phi^{\alpha}_{t}, \mathbf{F}).
\end{equation}

\begin{table}[!t]
    \centering
    \caption{Logical Regularizations for System~2}
    \vspace{-0.2cm}
    \label{Tab:Logic}
    \begin{tabular}{c l l}
    \hline
    \hline
    \textbf{Operation}&\textbf{Logical Rule} & \textbf{Logical Equation} \\
    \hline
    \multirow{1}*{$\neg$}
    &\textbf{Double Negation} & $\neg (\neg w) = w$ \\
    \hline
    \multirow{4}*{$\land$}&\textbf{Identity} & $w \land \mathbf{T} = w$ \\
    &\textbf{Annihilator} & $w \land \mathbf{F} = \mathbf{F}$  \\
    &\textbf{Idempotence} & $w \land w = w$  \\
    &\textbf{Complementation} & $w \land \neg w = \mathbf{F}$ \\
    \hline
    \multirow{4}*{$\lor$}&\textbf{Identity} & $w \lor \mathbf{F} = w$ \\
    &\textbf{Annihilator} & $w \lor \mathbf{T} = \mathbf{T}$ \\
    &\textbf{Idempotence} & $w \lor w = w$ \\
    &\textbf{Complementation} & $w \lor \neg w = \mathbf{T}$ \\
    \hline
    \multirow{5}*{$\rightarrow$}&\textbf{Identity} & $w \rightarrow \mathbf{T} = \mathbf{T}$  \\
    &\textbf{Annihilator} & $w \rightarrow \mathbf{F} = \neg w$ \\
    &\textbf{Idempotence} & $w \rightarrow w = \mathbf{T}$  \\
    &\textbf{Complementation} & $w \rightarrow \neg w \equiv \neg w$ \\
    \hline
    \hline
    \end{tabular}
    \vspace{-0.6cm}
\end{table}

\begin{figure*}[!t]
\begin{equation}
    \label{eq:bound}
\begin{aligned}
\ln \tilde{p}_\varphi\left(o_{1: T} \mid a_{1: T}\right) \ge & \sum_{t=1}^T(\underbrace{\mathbb{E}_{q_1}\left[\ln p_\varphi\left(o_t \mid z_t\right)\right]}_{\mathrm{Decoding \ Loss}} + \underbrace{\mathbb{E}_{q_1}\left[\ln \phi^{\alpha}_{t} \right]}_{\mathrm{Logic \ Limit}} -\underbrace{\mathbb{E}_{q_2}\left[ \mathrm{D}_{\textrm{KL}}\left[q_\varphi\left(z_t \mid o_{\leq t}, a_{<t}\right) \| p_\varphi\left(z_t \mid z_{t-1}, a_{t-1}\right)\right]\right]}_{\mathrm{Prediction \ Loss}}).
\end{aligned}   
\end{equation}
\hrulefill
\vspace{-0.5cm}
\end{figure*}

\noindent Let $\zeta=\{\boldsymbol{W}^l_{1},\boldsymbol{W}^l_{2},\boldsymbol{b}^l_{1},\boldsymbol{b}^l_{2}\}$ be the parameter set of the System~2 component. The total loss function of System~2 will be $\mathcal{L}_{\text{S2}}(\zeta) = \mathcal{L}^{\alpha}_{\text{log}} + \beta \mathcal{L}_{\text{reg}}$,
where $\beta \in (0,1)$ is the weight factor. To ensure the order-independence, i.e., $b \land a = b \land a$ and $b \lor a = b \lor a$, the order of inputs for AND and OR is randomly set during both offline training and online testing. 

\subsection{Inter-System Signal Mechanism}
The integration of System~1 and System~2 is essential to combine fast statistical inference with deep logical reasoning. To realize it, we enable interaction between System~1 and System~2 by an inter-system signal mechanism. Particularly, as shown in Fig.~\ref{Fig:System}, the System~1 component provides the System~2 component with actual observations of the wireless network. These observations serve as the labeled data, based on which the System~2 component learns the logical relationships of wireless network state transitions. Particularly, real-world trajectories $\mathcal{J} = \{s_{1:t},a_{1:t},r_{1:t}\}$ from System~1 are fed into System~2 to minimize the loss function (15), where $s_{t}=\left\{z_t,h_t\right\}$ captures both the stochastic and deterministic states of the wireless network. As shown in Fig.~\ref{Fig:Imag}, the logical consistency loss from the System~2 component guides the System~1 component during imagination, where the predictions of the latent network representations must follow the logical consistency.
Based on this process, we propose the \emph{logic-enhanced conditional latent-variable model} as follows.

\vspace{-0.25cm}
\begin{definition}
  We define a logic-enhanced conditional latent-variable model for the RSSM-based System~1 component with logical consistency, which is given by
$
  \tilde{p}_\varphi(o_{1:T}, z_{1:T} \! \mid \! a_{1:T}) \! = \! \prod_{t=1}^T \! p_\varphi(o_t \! \mid \! z_t) p_\varphi(z_t \! \mid \! z_{t-1}, a_{t-1}) \phi^{\alpha}_{t}. 
  $
\end{definition}
\vspace{-0.4cm}

\begin{theorem}
  Considering the logical signals from the System~2 component to the System~1 component, the LE-ELBO of the imagination loss can be bounded by (\ref{eq:bound}), where $q_1=q_\varphi\left(z_t \mid o_{\leq t}, a_{<t}\right)$, $q_2=q_\varphi\left(z_{t-1} \mid o_{\leq t-1}, a_{< t-1}\right)$, and the prior state is approximately obtained by $q_\varphi(z_{1:T} \mid o_{1:T}, a_{1:T}) = \prod_{t=1}^{T} q_\varphi(z_t \mid h_t, o_t)$.
 \end{theorem}
\begin{proof}
  See Appendix A.
\end{proof}
\vspace{-0.2cm}
Theorem 1 establishes a principled connection between logical reasoning and variational imagination by showing how the logical consistency of System~2 to tighten the ELBO bound of System~1's prediction loss over extended horizons. It enables more reliable and logically consistent trajectory predictions, which are essential for long-horizon planning in highly dynamic and complex wireless networks. 

The proposed inter-system signaling mechanism enables structured coordination between pattern-based prediction and logic-based correction. During training, the System~1 component provides latent wireless states extracted from real-world network observations, which serve as the foundation for the System~2 component to learn underlying logical relationships. During imagination, the System~2 component imposes logical constraints on the System~1 component to enable long-horizon, reliable predictions for the wireless network.

\vspace{-0.3cm}
\subsection{Learning Link Scheduling in Imagined Trajectories}
As shown in Fig.~\ref{Fig:Imag}, the imagination ability of the proposed dual-mind world model is used to simulate the future stochastic state $\{\tilde{z}_t\}$ of the wireless network for policy learning. It is data efficient since the policy is learned in the imagined trajectories without relying on high-cost actual interactions and real-time feedback from the real-world wireless network, as are the cases in RL. Moreover, the differentiable imagination provides long-horizon predictions to evaluate the current policy and attributes the delayed returns back to earlier actions, thus a long-term policy can be learned. Particularly, the predicted stochastic state is recurrently obtained by $\tilde{z}_t \sim p_\varphi\left(\tilde{z}_t \mid h_t\right)$ and $h_t = f_\varphi\left(h_{t-1}, \tilde{z}_{t-1}, \tilde{a}_{t-1}\right)$. Then, an imagined trajectory of the wireless network can be formulated as $\tilde{\mathcal{J}}_{t-1}=\left\{\tilde{s}_{t:t+H},\tilde{a}_{t:t+H},\tilde{r}_{t:t+H}\right\}$, where the state $\tilde{s}_{t}=\left\{\tilde{z}_t,h_t\right\}$ encodes the wireless data at timeslot $t$, and $H$ represents the horizon size of imagination.

Let $\mathcal{S}$ be the state space and $\mathcal{A}$ be the action space.
We apply the actor-critic framework as the planner to learn link scheduling in imagined trajectories of the wireless network. Particularly, the actor-critic model involves two components: the actor component $\tilde{a}_{\tau} \sim q_{\theta}\left(\tilde{a}_{\tau} \mid \tilde{s}_{\tau}\right)$ for policy learning and the critic component $v_\psi(\tilde{s}_{\tau}) \approx \mathbb{E}_{q(\cdot \mid \tilde{s}_{\tau})} \left( \sum_{t=\tau}^H \gamma^{t-\tau} \tilde{r}_{\tau} \right)$ for state value estimation,
where $\psi$ represents the parameter of the critic, and $\theta$ represents the parameter of the actor.
With the imagined trajectory $\tilde{\mathcal{J}}$, the actor learns to maximize the return value by link scheduling, and the critic learns to evaluate the long-term return from CAoI. Hence, the actor and the critic can be respectively optimized by
\vspace{-0.2cm}
\begin{equation}
  \begin{aligned}
    &\theta^* = \max_\theta \mathbb{E}_{q_\phi, q_\theta} \left[ \sum_{\tau=t}^{t+H} V_\lambda(\tilde{s}_{\tau}) \right] ,\\
    &\psi^* = \min_\psi \mathbb{E}_{q_\phi, q_\theta} \left[ \sum_{\tau=t}^{t+H} \frac{1}{2} \left( v_\psi(\tilde{s}_{\tau}) - V_\lambda(\tilde{s}_{\tau}) \right)^2 \right].
  \end{aligned}
\end{equation}
To evaluate the long-term performance of the network, the value $V_\lambda(\tilde{s}_{\tau})$ with discount weight $\lambda$ is given by 
\vspace{-0.2cm}
\begin{equation}
  \begin{aligned}
    &V_\lambda(\tilde{s}_{\tau}) = (1 - \lambda) \left( \sum_{n=1}^{H-1} \lambda^{n-1} V_n^N(\tilde{s}_{\tau}) \right) + \lambda^{H-1} V_H^N(\tilde{s}_{\tau}), \\
    &V_k^N(\tilde{s}_{\tau}) = \mathbb{E}_{q_\phi, q_\theta} \left[ \sum_{n=\tau}^{h-1} \gamma^{n-\tau} \tilde{r}_{n} + \gamma^{h-\tau} v_\psi(\tilde{s}_{h}) \right],
  \end{aligned}
\end{equation}
where $h = \min(\tau + k, t + H)$. Moreover, the actual reward $r_t$ of the network during the timeslot $t$ is designed as
$
  r_t = - \frac{1}{V} \sum_{v}\left[ A^{v}_t - \mathbb{I}(A^v_t > \bar{A}) (\bar{A}-A^v_t)\right]
$.

In a real-world mmWave V2X network, it is difficult and inefficient to obtain the real-time wireless data $\{o_t\}$ for each extremely short timeslot $t$ when the size of wireless data $\{o_t\}$ is large. 
In this context, the imagination ability of the world model can be leveraged for joint prediction of the wireless data and the link scheduling without real-time data collection in practical applications, as illustrated in Fig.~3. Given the deterministic trajectory $\mathcal{J}[c]=\left\{s[1:c],a[1:c],r[1:c]\right\}$ that is collected from actual scenario over $c$ deterministic timeslots, the world model can predict a trajectory $\tilde{\mathcal{J}}[c]=\left\{\tilde{s}[c+1:c+Y],\tilde{a}[c+1:c+Y],\tilde{r}[c+1:c+Y]\right\}$ for a few, future timeslots $Y$. It is practical for real-world wireless applications. For instance, in the absence of real-world observations $\{o[c+1:c+Y]\}$ from wireless sensors, the world model can infer the future state representations $\tilde{\mathcal{J}}[c]$ of the wireless network from historical observations. These predicted states serve as imagined environments for planning the upcoming link scheduling actions $a[c+1:c+Y]$.

In a nutshell, the proposed dual-mind world model-based learning approach addresses the CAoI minimization problem in (4) by jointly learning statistical pattern-driven System~1 that captures the dynamics of wireless data, and logic-driven System~2 that recognizes the logical relationships of the wireless network state transitions. Then, a long-term policy is learned in long-horizon imagined trajectories with logical consistency. Hence, the proposed dual-mind world model approach addresses the following challenges: (a) It provides more reliable imagined trajectories for wireless networks to alleviate the accumulated prediction errors over an extended horizon compared to independent System~1, and ensures logical consistency of imagination even with unseen states, (b) It can easily attribute delayed rewards back to earlier link scheduling since the imagination is differentiable, (c) It is highly data-efficient since the link scheduling is trained in imagination instead of real-world network interactions and real-time wireless data, and (d) It ensures wireless networks can learn the long-term planning since imagined trajectories provides foresight returns of policies over a long horizon $H$.

Here, we define necessary notations to better introduce and analyze the proposed dual-mind world model as follows.
Let the training episodes be $N^{\mathrm{tra}}$, seed episodes be $N^{\mathrm{seed}}$, batch size be $\Theta$, sequence length be $L$, replay buffer be $\mathcal{D}$, collect interval be $N^{\mathrm{col}}$, and the learning rates of parameters $\vartheta$, $\psi$, $\phi$, and $\zeta$ respectively be $\rho_{\vartheta}$, $\rho_{\psi}$, $\rho_\phi$, and $\rho_\zeta$. The training process of the proposed dual-mind world with the actor-critic-based planner for wireless networks is summarized in Algorithm \ref{Alg:WMT}, and the practical use case without real-time available wireless data is summarized in Algorithm \ref{Alg:WMT2}. The overall training complexity of the proposed approach is $O\left(N^{\operatorname{tra}}\times N^{\mathrm{col}}(\Theta L+H+\alpha)C\right)$, which corresponds to a one-step forward-backward computation per training iteration. This complexity scales linearly with the number of training episodes and collected samples, and is comparable to that of standard model-based reinforcement learning methods, while providing improved data efficiency through imagination-based policy learning. In actual deployments, the proposed dual-mind model provides rapid inference through the System~1 component without the need of extra inference overhead from the System~2 component. Hence, the proposed dual-mind world model can be used in wireless networks with low latency computing requirements. Moreover, the objective of the world model is to learn and predict the dynamics of the wireless network, and to construct a foresighted planner that learns a near-optimal scheduling policy for the CAoI minimization problem (\ref{P:CAoI}), rather than solving it in a closed-form manner.

\setlength{\textfloatsep}{1pt} 
\setlength{\intextsep}{1pt}    
\begin{algorithm}[H]
    \caption{Proposed Dual-Mind World Model With Actor-Critic-Based Planner for Wireless Networks}
    \label{Alg:WMT}
    \scriptsize
 \begin{algorithmic}
    \STATE Initialize the wireless network, and $\mathcal{D}$ with $N^{\mathrm{seed}}$ episodes.
    \FOR{Training episode $n^{\mathrm{tra}} \rightarrow N^{\mathrm{tra}}$}
    \FOR{Collect interval $n^{\mathrm{col}} \rightarrow N^{\mathrm{col}}$}
    \STATE // Learn Network Patterns By System~1
    \STATE Sample $\Theta$ sequences $\{(o_t, a_t, r_t)\}_{t=k}^{k+L} \sim \mathcal{D}$.
    \STATE Predict prior $\tilde{z}_t$, $\tilde{r}_t$ with $h_t$, and decode $\hat{o}_t$.
    \STATE Update RSSM $\phi \gets \phi - \rho_\phi \nabla_\phi \mathcal{L}_{\text{S1}}(\phi)$.
    \STATE $\triangleright$ Learn Network Logical Rules By System~1
    \STATE Self-supervised learn logic rules by $\mathcal{L}_{\text{reg}}$.
    \STATE Learn logic from System~1's network states by $\mathcal{L}_{\text{log}}$.
    \STATE Update LINN $w \gets \zeta - \rho_\zeta \nabla_\zeta \mathcal{L}_{\text{S2}}(w)$.
    \STATE $\triangleright$ Train Actor-Critic Based Planner In Imagination
    \STATE Act in imagination $\{(\tilde{z}_{\tau}, a_\tau)\}_{\tau=t}^{t+H}$ from actual $z_t$.
    \STATE Estimate value $V_\lambda(s_\tau)$ with imagined rewards $\{\tilde{r}_{\tau}\}$.
    \STATE  $\vartheta \gets \vartheta + \rho_{\vartheta} \nabla_\vartheta \sum_{\tau=t}^{t+H} V_\lambda(s_{\tau})$.
    \STATE  $\psi \gets \psi - \rho_{\psi} \nabla_\psi \sum_{\tau=t}^{t+H} \frac{1}{2} \| v_\psi(s_{\tau}) - V_\lambda(s_{\tau}) \|^2$.
    \STATE $\triangleright$ Logical Rules From System~2 to System~1
    \STATE Ensure logic consistency of $\{(\tilde{z}_{\tau}, a_{\tau})\}$ by $\mathcal{L}_{\text{log}}$
    \STATE $\triangleright$ Differentiable Feature of Imagination
    \STATE Update RSSM $\psi \gets \psi - \rho_\psi \nabla_\psi \mathcal{L}_{\text{S2}}(\psi)$.
    \ENDFOR
    \STATE Reset environments of the wireless network.
    \FOR{Time step $t \rightarrow T$}
    \STATE Obtain $h_t$ and $z_t$ from $o_t$ by System~1.
    \STATE Plan $a_t \sim q_\vartheta(a_t \mid z_t)$ and act in the network.
    \ENDFOR
    \STATE Add experience to buffer $\mathcal{D} \gets \mathcal{D} \cup \{(o_t, a_t, r_t)\}_{t=1}^T$.
    \ENDFOR
    \STATE Return $\phi^*$, $\zeta^*$, $\vartheta^*$ and $\psi^*$.
 \end{algorithmic}
 \end{algorithm}
 
 \begin{algorithm}[H]
    \caption{Practical Use Case of Joint Prediction And Link Scheduling without Real-Time Available Wireless Data}
    \label{Alg:WMT2}
    \scriptsize
 \begin{algorithmic}
    \STATE Deploy the trained world model with $\phi^*$, $w^*$, $\vartheta^*$ and $\psi^*$.
    \FOR{Time step $t \rightarrow T$}
    \IF{Obtain wireless data at timeslot  $t$}
    \STATE Obtain $h_t$ and $z_t$ from $o_t$ by System~1.
    \STATE Plan $a_t \sim q_\vartheta(a_t \mid z_t)$ and act in real world.
    \ELSE
    \STATE Imagine $\tilde{h}_t$ and $\tilde{z}_t$ from $h_{\leq t-1}$ by System~1.
    \STATE Plan $a_t \sim q_\vartheta(a_t \mid \tilde{z}_t)$ and act in real world.
    \ENDIF
    \ENDFOR
 \end{algorithmic}
\end{algorithm}

\begin{figure}[t!]
	\begin{minipage}[b]{\columnwidth}
		\centering
		\subfloat[The Flushing Avenue of New York on OpenStreetMap from the ArcGIS satellite.]{\includegraphics[width=0.8\linewidth]{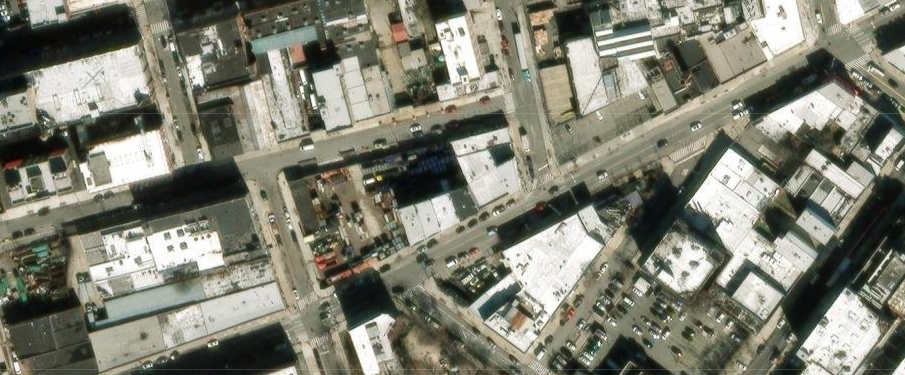}}
    \vspace{-0.15cm}
    \\
		\subfloat[Blender-based 3D scenario creation and rendering for a real-world environment importing from OpenStreetMap.]{\includegraphics[width=0.8\linewidth]{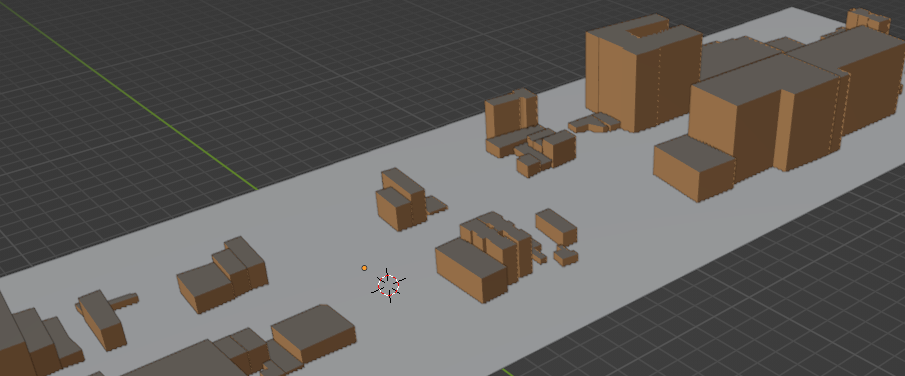}}
    \vspace{-0.2cm}
    \\
		\subfloat[Sionna-based simulator with mobile vehicles for realistic LoS links, specular reflection, diffuse reflection, and refraction.]{\includegraphics[width=0.8\linewidth]{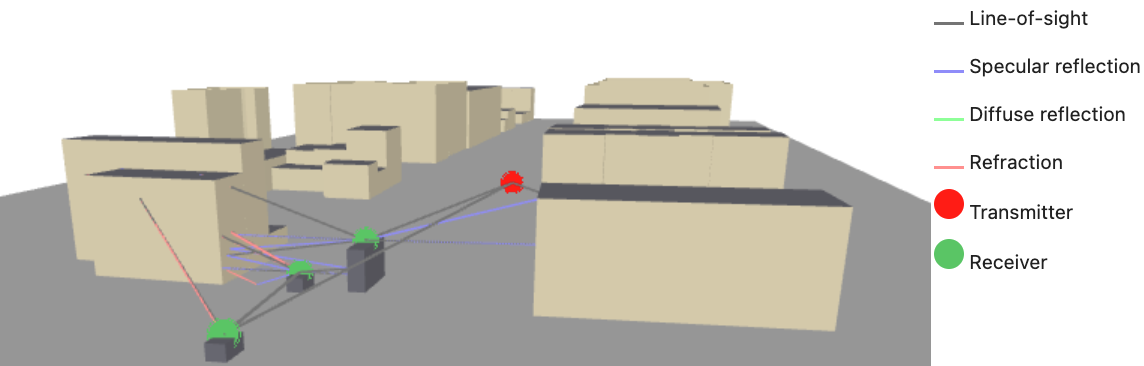}}
    \vspace{-0.05cm}
	\end{minipage}
	\caption{Procedures of the proposed realistic simulator based on Sionna, Blender, ArcGIS, Mitsuba, and plug-in of Blender-OSM and Mitsuba-Blender.}
  \label{Fig:Sio}
  \vspace{-0.2cm}
\end{figure}

\vspace{-0.5cm}
\section{Simulation and Analysis}
\vspace{-0.2cm}
\subsection{Realistic Sionna-based Simulator}
\vspace{-0.2cm}
For our simulations, we develop a novel realistic simulator based on Sionna, Blender, ArcGIS, Mitsuba, and plug-in of Blender-OSM and Mitsuba-Blender~\cite{sionna}, as shown in Fig.~\ref{Fig:Sio}. Particularly, we first select the urban scenario on OpenStreetMap, where we choose the Flushing Avenue of New York, as shown in Fig.~\ref{Fig:Sio}(a). Then, as shown in Fig.~\ref{Fig:Sio}(b), we load the selected scenario into Blender, which is an industrial 3D rendering suite, by the plug-in of Blender-OSM to create the Mitsuba files. Finally, as shown in Fig.~\ref{Fig:Sio}(c), the Mitsuba files are imported to Sionna to create a realistic physical scenario, and the application programming interfaces provided by Sionna are invoked to simulate the real-world signal propagation over the links of LoS, specular reflection, diffuse reflection, and refraction. The mobile vehicles are generated by using the Mitsuba-Python tool and dynamically added to the constructed physical scenario.
Based on the proposed Sionna-based realistic simulator, we generate a realistic urban mmWave V2X scenario, which provides the physics-enhanced end-to-end channel models and ray-tracing data along with different material properties of scene objects. 
The ray tracing data serves as the real-world physical channel data $\Xi$, which consists of the delay of multipath, azimuth and zenith angles of departure (AoD), azimuth and zenith angles of arrival (AoA), time of departure (ToD), and the time of arrival (ToA). Here, we only select the strongest path of all links for each vehicle to characterize the channel propagation features.

\begin{table}[!t]
  \centering
  \caption{Hyperparameters}
  \vspace{-0.2cm}
  \begin{tabular}{|l|c|c|}
  \hline
  \textbf{Parameter} & \textbf{Symbol} & \textbf{Value} \\ \hline
  \multicolumn{3}{|l|}{\textbf{Environment}} \\ \hline
  Number of vehicles & $\Psi$ & 8 \\ 
  Packet size & $S$ & 5 $~\mathrm{MB}$ \\
  Number of packets & $C^u$ & 25  \\
  Bandwidth & $B$ & 100 $~\mathrm{MHz}$ \\ 
  Frequency & $f_c$ & 26 $~\mathrm{GHz}$ \\ 
  Transmit power & $P_v$, $P_u$ & 23 $~\mathrm{dBm}$ \\
  Timeslot duration & $\xi$ & 100 $~\mathrm{ms}$ \\
  Period & $T$ & 100 \\ 
  Number of antennas & --- & 4  \\
  CAoI tolerance & $\bar{A}$ & 8 \\
  Vehicle speed & --- & 15-20 $~\mathrm{m/s}$\\ 
  Vehicle security distance & --- & 20 $~\mathrm{m}$\\ \hline
  \multicolumn{3}{|l|}{\textbf{Proposed dual-mind world model framework}} \\ \hline
  Seed episode & $N^{\mathrm{seed}}$ & 5 \\ 
  Sequence length & $L$ & 64 \\ 
  Training episodes & $N^{\mathrm{tra}}$ & 1e3 \\
  Collect interval & $N^{\mathrm{col}}$ & 100 \\
  Replay buffer size & $|\mathcal{D}|$ & 1e6 \\ 
  Batch size & $\Theta$ & 50 \\ 
  Imagination horizon & $H$ & 30 \\ 
  Stochastic state size & $|z_t|$, $|\tilde{z}_t|$ & 256 \\ 
  Deterministic state size & $|h_t|$, $|h_t|$ & 256 \\ 
  Activation layer function & --- & Relu \\ 
  Loss weights & $\delta_{_{\text {dyn}}}$, $\delta_{_{\text {rep}}}$ & 1 \\ 
  Reasoning depth & $\alpha$ & 30 \\ 
  Logic vector size & $|v|$, $|m|$ & 64 \\
  System 1 optimizer & --- & Adam ($\epsilon$ = 1e-4) \\ 
  System 1 learning rate & $\rho_\phi$ & 1e-3 \\
  System 2 optimizer & --- & SGD ($\epsilon$ = 1e-4) \\ 
  System 2 learning rate & $\rho_\zeta$ & 1e-2 \\ 
  \hline
  \multicolumn{3}{|l|}{\textbf{Actor-critic for policy learning}} \\ \hline
  Exploration noise & --- & 0.3 \\
  Return lambda & $\lambda$ & 0.95 \\ 
  Planning horizon discount & $\gamma$ & 0.99 \\ 
  Actor-critic optimizer & --- & Adam ($\epsilon$ = 1e-4) \\ 
  Learning rate & $\rho_{\vartheta}$, $\rho_{\psi}$ & 1e-4 \\ 
  \hline
  \end{tabular} 
  \label{Tab:1}
\end{table}

\vspace{-0.5cm}
\subsection{Parameter Setup}
An urban road with 200 meter length and $\Upsilon=3$ parallel lanes is consider as a physical scenario without the lane-changing behavior of vehicles. For the proposed dual-mind world model, we use typical parameters as in \cite{hafner2023mastering} and \cite{wang2025dmwm}. All training is conducted on a NVIDIA RTX 4070 GPU, and the training of the proposed world model takes approximately 0.4 GPU days, not accounting for the time of ray-tracing data collection. All of the hyperparameters are presented in Table \ref{Tab:1}. For comparison, we benchmark the proposed dual-mind world model (DMWM) against state-of-the-art baselines including the model-free discrete soft actor-critic (MFRL-SAC) approach~\cite{wang2023adaptive}, the model-based policy optimization (MBRL-MBPO) approach~\cite{janner2019trust}, and our prior proposed world model (WM-System~1) that only considers System~1~\cite{wang2025world}. 

\vspace{-0.5cm}
\subsection{Data Efficiency}

Fig.~5 and Fig.~6 show the average test rewards over 100 test episodes under limited environment steps and under limited environment trials, respectively. The network steps represent the amount of wireless data used for training from the actual V2X network, and the environment trials refers to the number of network learning opportunities, where once the CAoI of the network exceeds the maximum tolerance, one learning opportunity ends. The measurement of environment trials can capture both practical learning opportunities and safety constraints, thus ensuring efficiency while preventing unsafe sample accumulation. As shown in Fig.~ 5 and Fig.~6, the proposed DMWM-based learning approach exhibits significantly improved data efficiency over the traditional RL methods and the existing world model methods. From Fig. 5, we can see that DMWM achieves a better performance with only $2 \times 10^5$ environment steps compared to $5 \times 10^6$, $1 \times 10^7$ and $5 \times 10^5$ environment steps required by MBPO, DSAC and WM-System~1, respectively. This is due to the fact that the world model decouples environment cognition from policy learning by constructing a predictive latent-space model of the network dynamics, accurately captures the dynamics and uncertainty of V2X networks with this predictive model, and learns long-term policies in differentiable imagined trajectories rather than a great number of environment interactions. In contrast, MFRL-SAC relies on excessive environment interactions due to its trial-and-error mechanism, and MBRL-MBPO cannot address the long-horizon error accumulation without reasoning and predictive latent-space state representations. Hence, world model-based learning approaches overcomes the low data efficiency of the existing learning approaches. Compared to the WM-System~1, the proposed DMWM can learn underlying logical rules from the network dynamics, thus achieving 2-fold improvement in data efficiency.

\begin{figure}[!t]
  \centering
  \vspace{-0.6cm}
  \includegraphics[width=0.8\linewidth]{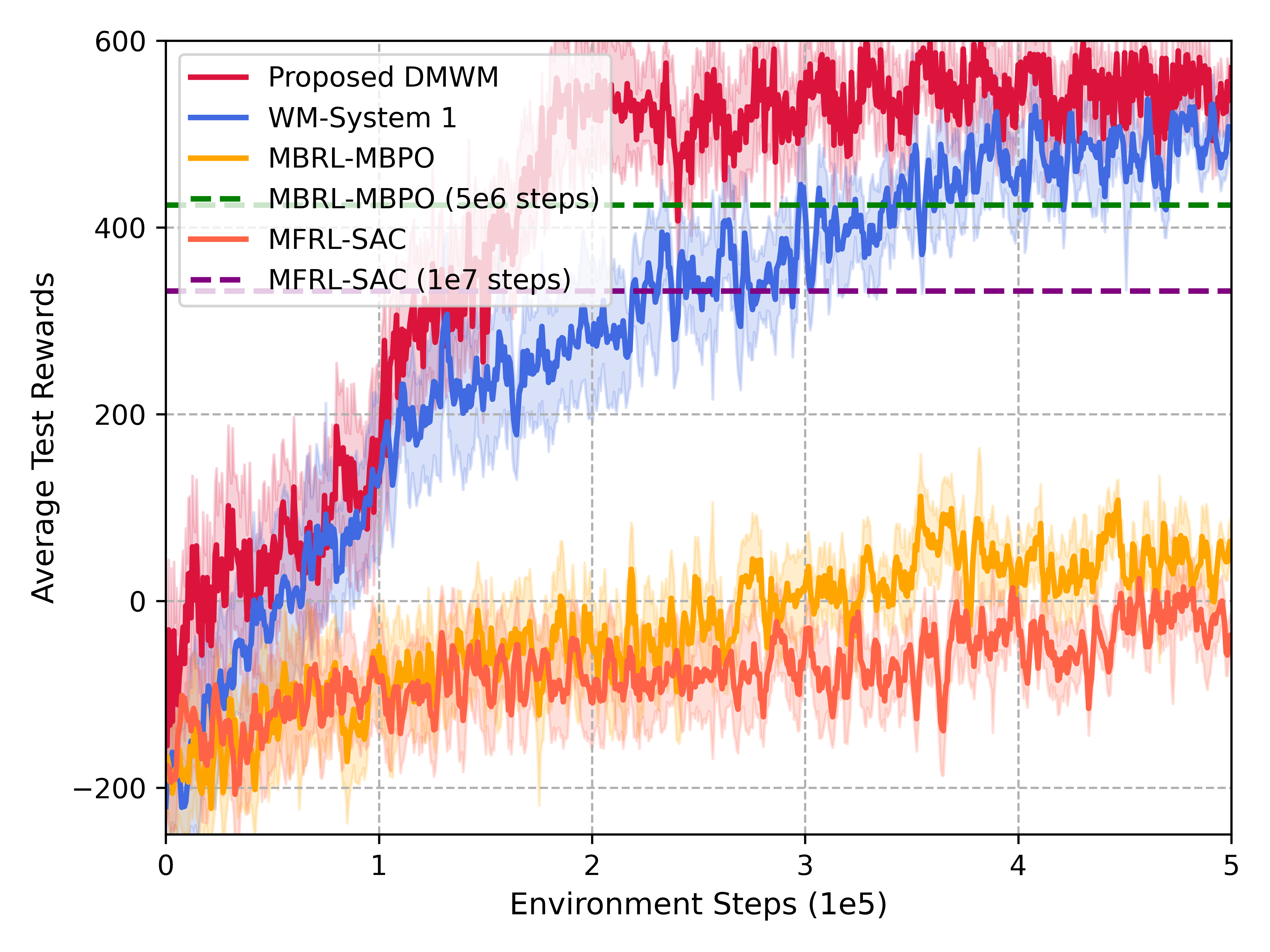}
  \vspace{-0.3cm}
  \caption{The average test rewards of different schemes under limited environment steps.}
  \vspace{-0.5cm}
\end{figure}

\begin{figure}[!t]
  \centering
  \includegraphics[width=0.8\linewidth]{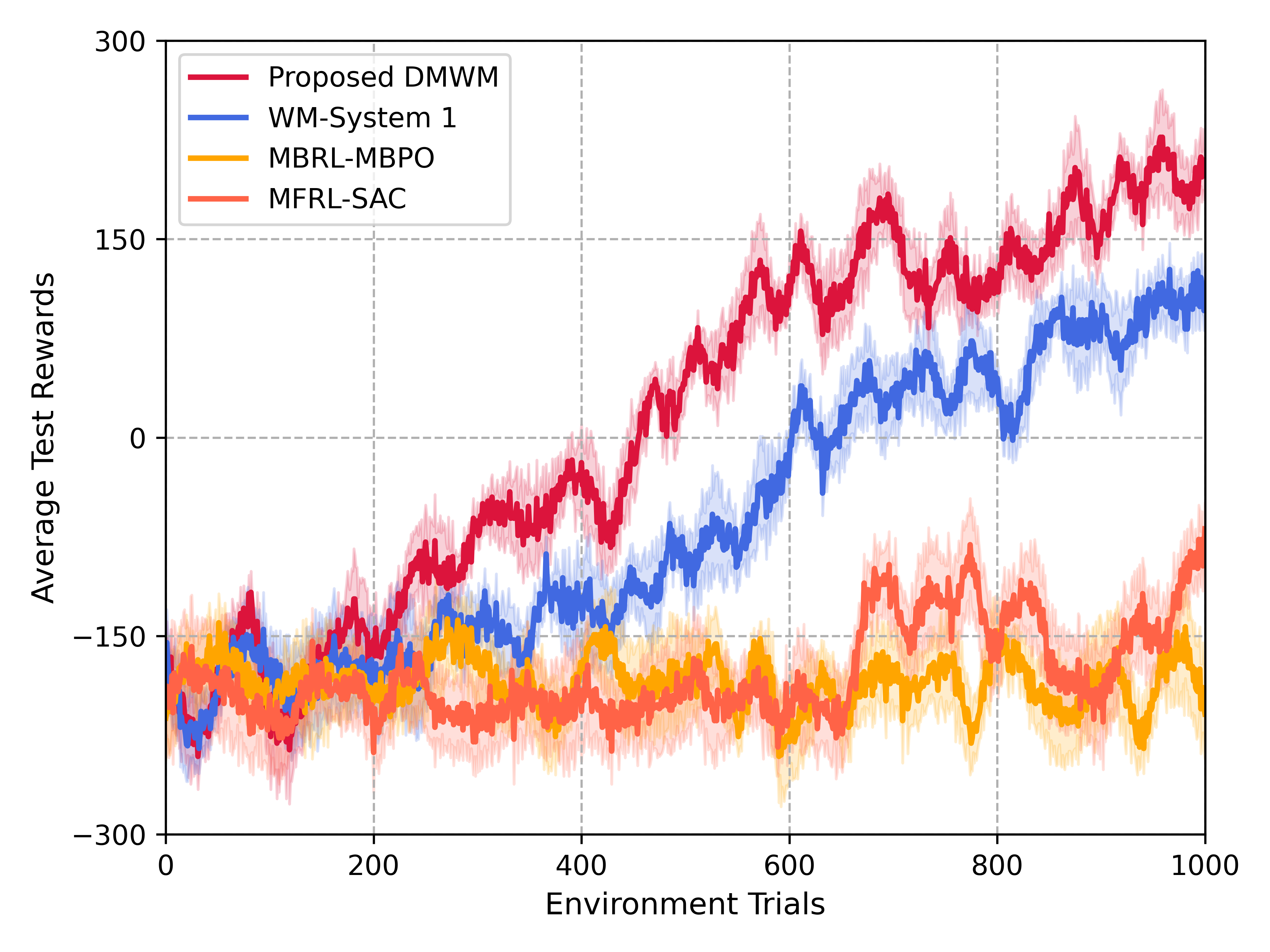}
  \vspace{-0.4cm}
  \caption{The average test rewards of different schemes under limited environment trials.}
  \vspace{-0.2cm}
\end{figure}

\begin{figure}[!t]
  \centering
  \vspace{-0.8cm}
  \includegraphics[width=0.85\linewidth]{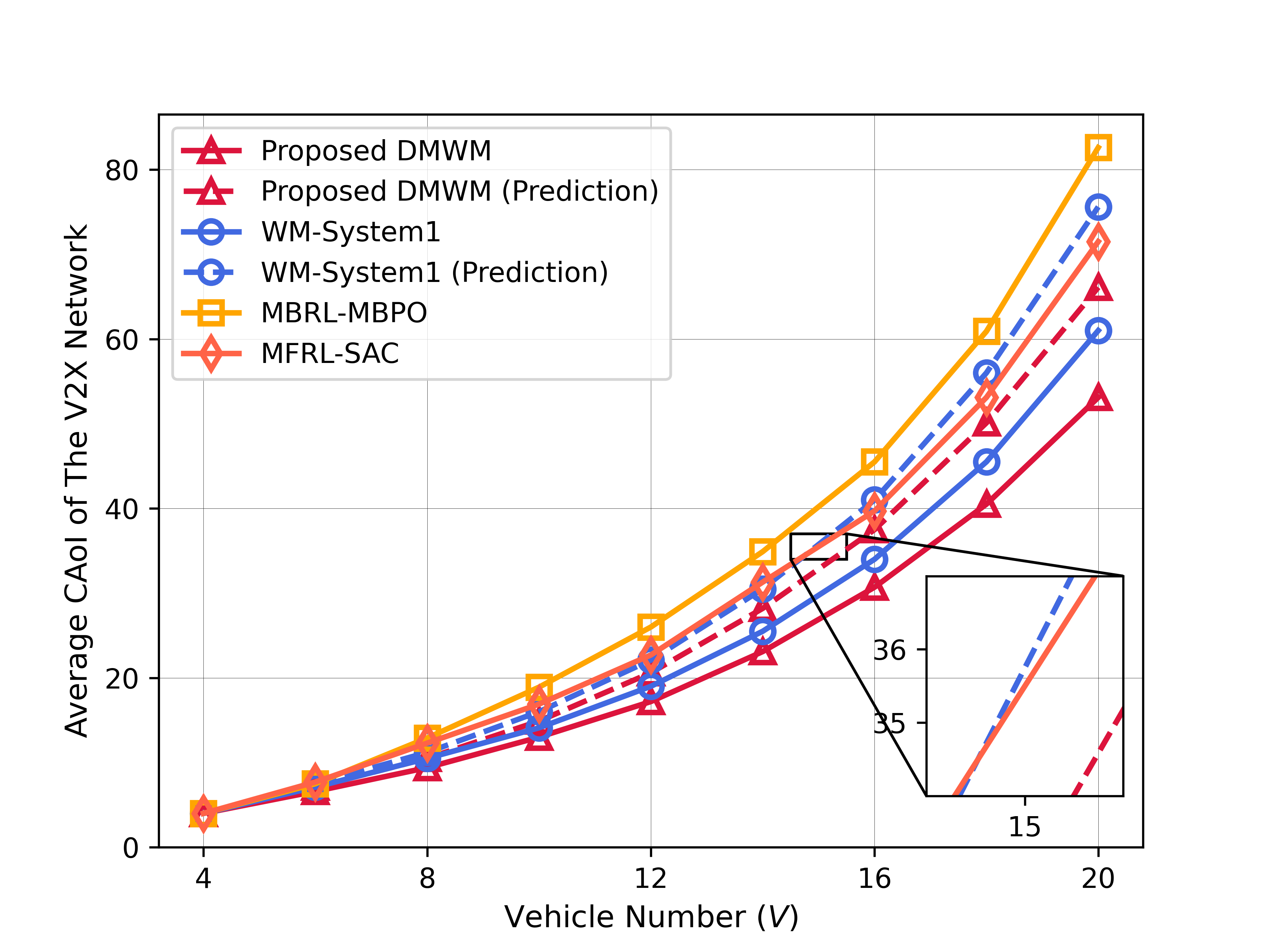}
  \vspace{-0.3cm}
  \caption{The average CAoI of the V2X network versus different number of vehicles.}
\end{figure}

\vspace{-0.3cm}
\subsection{Performance Comparison}
Fig.~7 shows the average CAoI of the V2X network versus different numbers of vehicles. 
DMWM improves the CAoI by up to 16\%, 32\% and 22\%, respectively, compared to MBRL-MBPO, MFRL-SAC, and WM-System~1. This is due to the long-term planning ability of DMWM that jointly considers long-horizon CAoI states of vehicles with logical consistency and the reliability of link scheduling, thus selecting the optimal solution over a long horizon. It is also observed that the DMWM and WM-System~1 with only imagined states, named ``Proposed DMWM (Prediction)" and ``WM-System~1 (Prediction)", respectively, can maintain stable performance close to RL approaches with real-time wireless data unavailable, which is significant for practical deployment and applications with occasional data interruptions.

Fig. 8 studies the long-term prediction performance versus different logical inference depths $\alpha$ of the System~2 component. From Fig. 10, we observe that a moderate logical inference depth can significantly reduce CAoI at larger prediction steps since the System~2 component ensures the multi-step consistency across imagined transitions of the V2X network. However, the gain from increasing logical inference depth will saturate. This is because a relatively small depth is sufficient to model the physics and dynamics of blockage, mobility, and scheduling. In this context, longer implication chains can compound rollout errors and propagate noise in symbolic predicates across steps. Hence, the optimal logical inference depth is scenario-dependent that different wireless environments exhibit distinct complex temporal dependencies and error accumulation. Moreover, we must mention that the additional logical depth is used to regularize imagination and training, while the online execution still relies on System 1 with online runtime latency remaining unchanged.

Fig. 9 shows the average CAoI versus imagination horizon for different network sizes. As the complexity of the network increases, the longer imagination horizon can improve the CAoI since it captures delayed endogenous CAoI feedback and mobility-induced exogenous dynamics. However, excessive horizon lengths can lead to accumulated model error, which in turn undermines the long-term planning ability. Compared to the world model with only System~1, DMWM achieves 26.1\% improvement when horizon size $H=40$ and 44\% improvement when horizon size $H=50$ with the number of vehicles $V=16$. This is because the System~2 component imposes long-horizon logical consistency that suppresses accumulated rollout error from the System~1 component.
Hence, the dual-mind approach is practically applicable to complex wireless networks that require robust planning over extended horizons.



\begin{figure}[!t]
  \centering
  \vspace{-0.55cm}
  \includegraphics[width=0.85\linewidth]{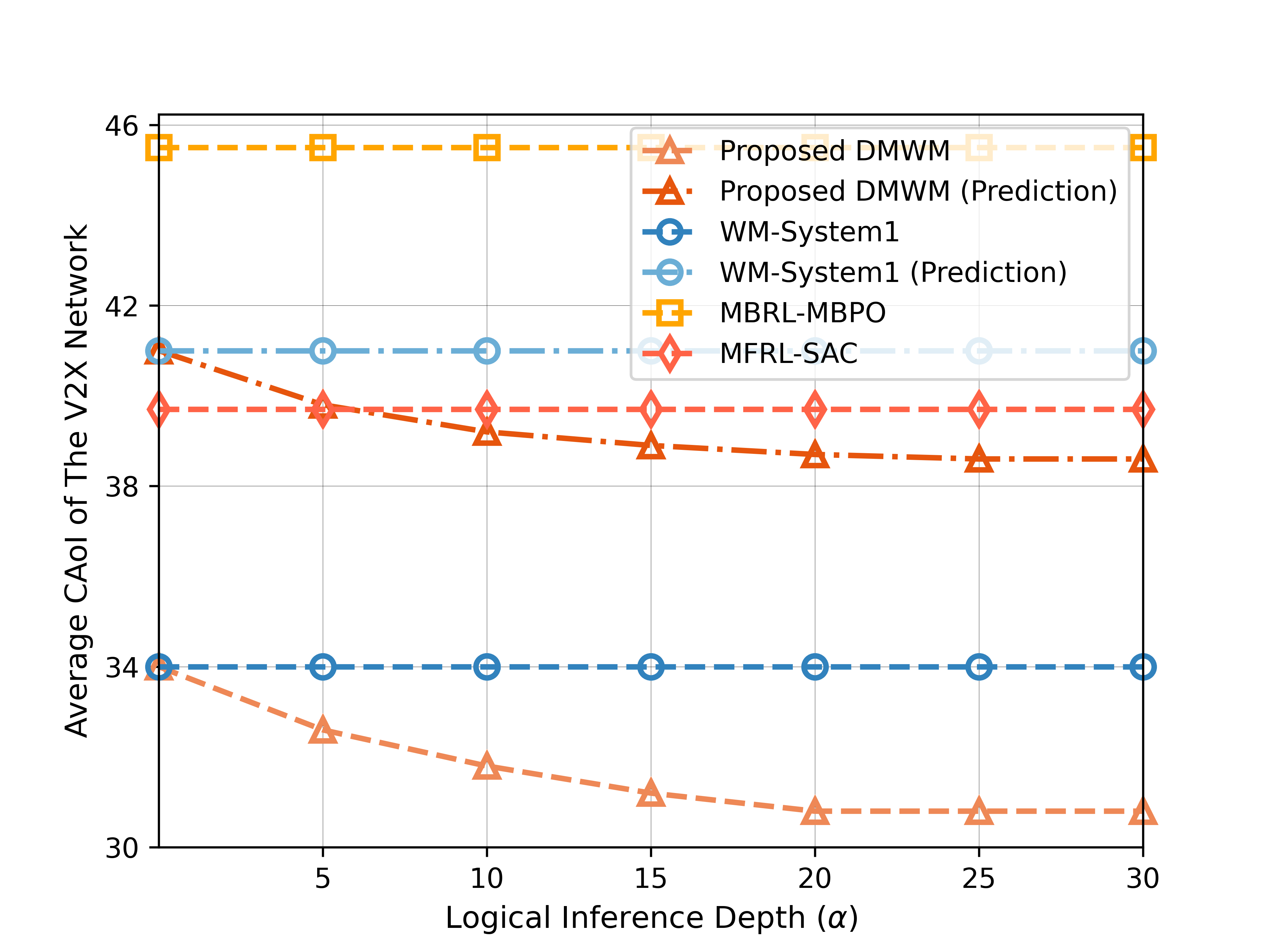}
  \vspace{-0.2cm}
  \caption{The average CAoI of the V2X network versus prediction steps with different logical inference depth.}
\end{figure}

\vspace{-0.3cm}
\subsection{Generalization}
\textbf{Observation and action masking:} The proposed DMWM for wireless learns to perform scheduling under varying numbers of vehicles, i.e.,  dynamic observation and action spaces. To adapt to the varying dimensions, we introduce a masking mechanism for zero-pad inputs and outputs to the largest respective dimensions \cite{hansen2023td} of the observation space and the action space. In practical V2X use cases, the maximum number of vehicles is decided by the limited maximum coverage range of an RSU. In particular, during the training and inference, we will mask out the invalid dimensions in predictions and actions. It ensures that prediction errors in invalid observation dimensions cannot influence the latent representation and policy learning. The link scheduling is sampled only along the valid action dimensions during planning.

\textbf{Generalization settings:} In Fig. 10, we consider new road scenes, unseen numbers of lanes and vehicles, and their joint dynamic combinations to simulate the dynamic environments, physics and network topology in real V2X networks. We evaluate both few-shot learning based on pretrained models and learning from scratch with limited samples. In Fig. 10(a), the models are trained with three scenarios and are generalized to three unseen scenarios with $\Psi = 12$ and $\Upsilon=3$ for scene generalization. In Fig. 10(b), the models are trained on $\Upsilon \in \{1,3,5\}$ and are generalized to unseen numbers of lanes $\tilde{\Upsilon}\in \{2,4,6\}$ with $\Psi = 12$ and a fixed wireless scenario for physical generalization. In Fig. 10(c), the models are trained with $\Psi \in \{10,12,14\}$ and are generalized to unseen numbers of vehicles $\tilde{\Psi}\in \{11,13,15\}$ with $\Upsilon=3$ and a fixed wireless scenario for network topology generalization. In Fig. 10(d), the models are trained on three scenarios with $\Psi \in \{10,12,14\}$, $\Upsilon \in \{1,3,5\}$ and are generalized to dynamic combinations of three unseen scenarios, $\tilde{\Upsilon}\in [1,6]$ and $\tilde{\Psi}\in [10,15]$.

\begin{figure}[!t]
  \centering
  \vspace{-0.2cm}
  \includegraphics[width=0.85\linewidth]{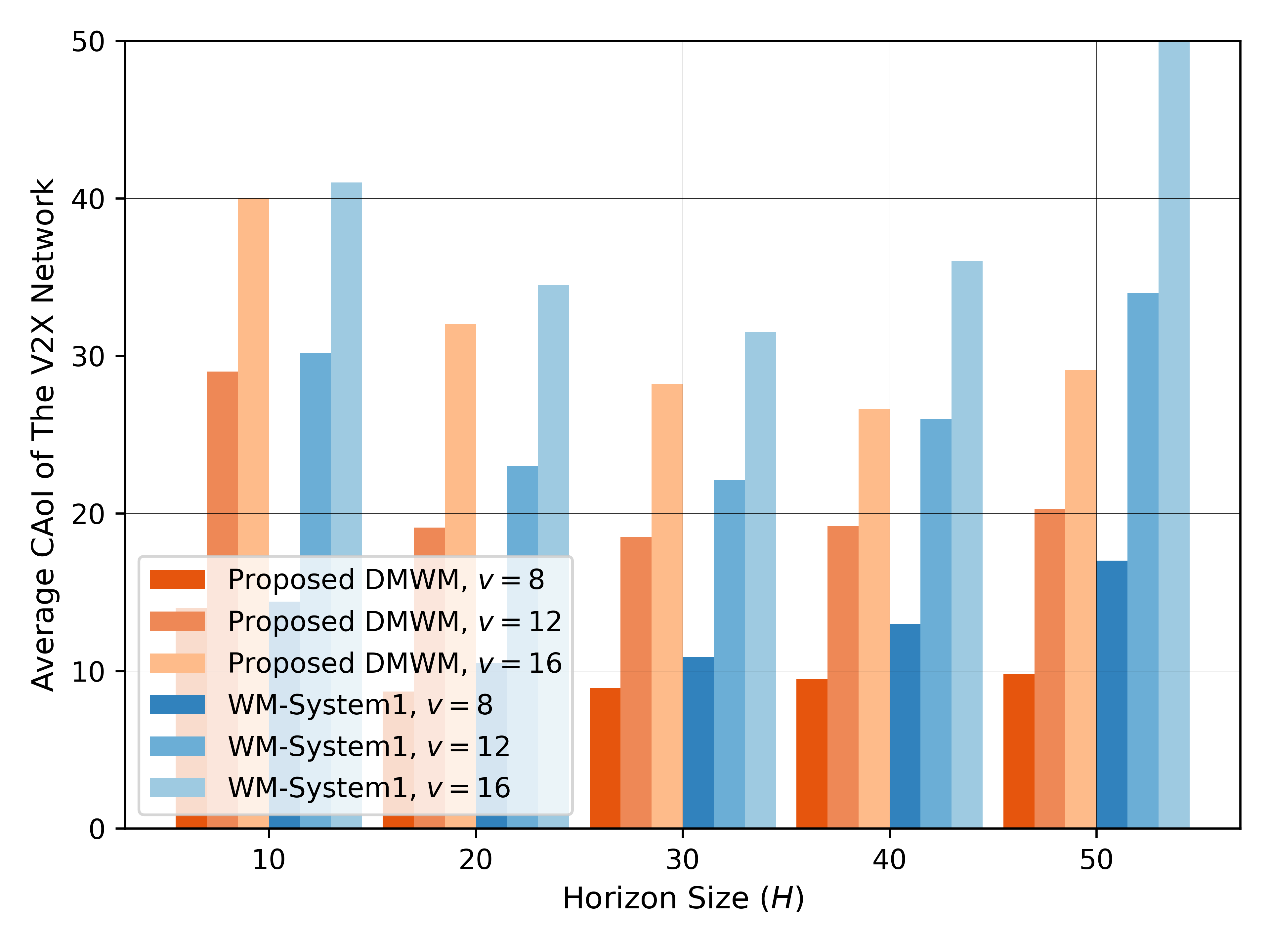}
  \vspace{-0.4cm}
  \caption{The average CAoI of the V2X network versus imagination horizon with different network complexity.}
\end{figure}

\begin{figure*}[!t]
  \centering
  \includegraphics[width=0.66\linewidth]{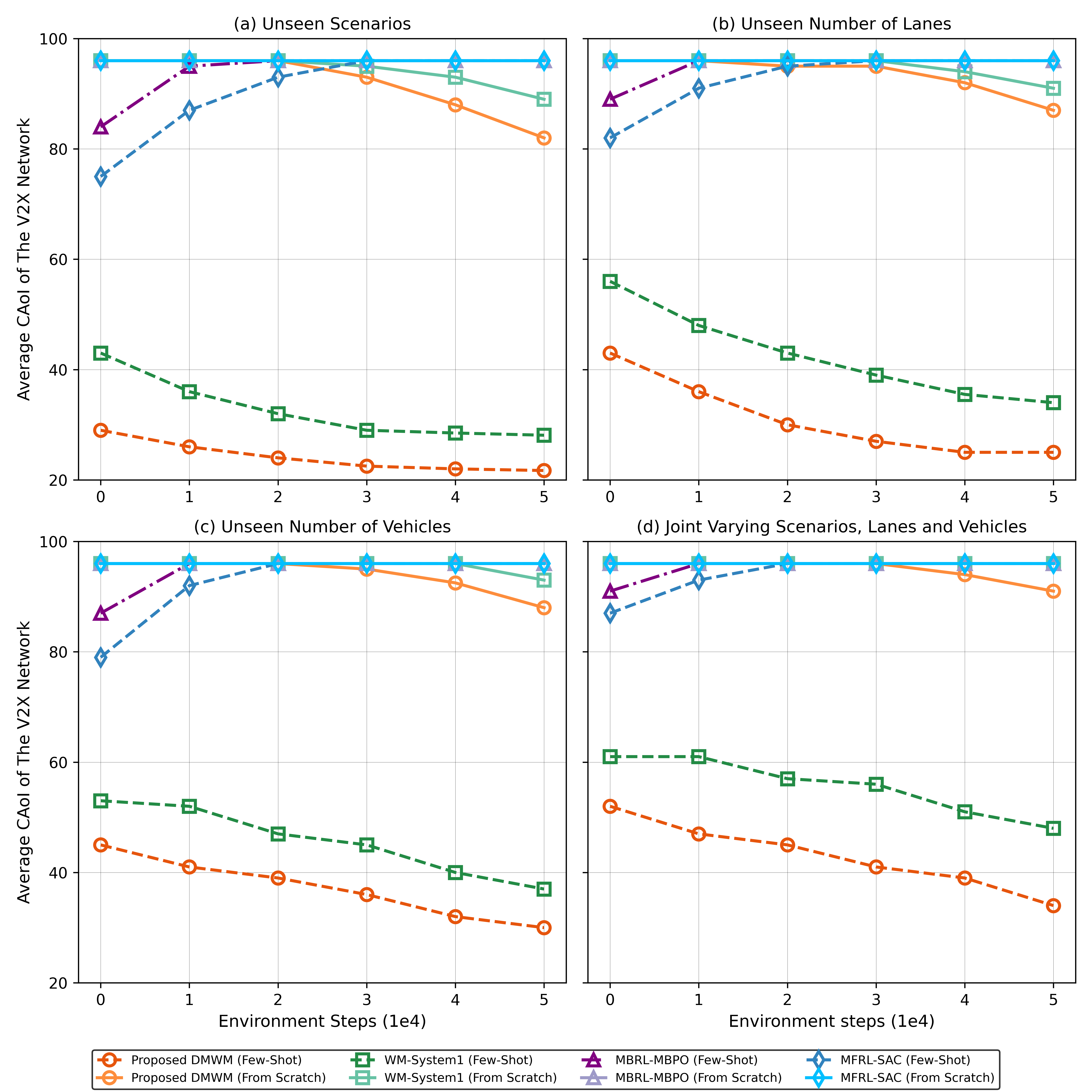}
\vspace{-0.1cm}
\caption{Generalization studies of different approaches adapting to (a) three unseen scenarios for scene generalization with pretraining on three given scenarios, (b) unseen numbers of lanes $\tilde{\Upsilon}\in \{2,4,6\}$ for physical generalization with pretraining on $\Upsilon \in \{1,3,5\}$, (c) unseen numbers of vehicles $\tilde{\Psi}\in \{11,13,15\}$ for network topology generalization with pretraining on $\Psi \in \{10,12,14\}$, and (d) jointly varying scenarios, number of lanes, and the number of vehicles.}
    \vspace{-0.6cm}
\end{figure*}

Fig. 10 shows the generalization and adaptation capability of different approaches to unseen environments beyond training data. Across all four generalization settings, DMWM achieves the best CAoI performance and adaptation with the fewest wireless data during few-shot learning.
In particular, compared to the world model with only System 1, DMWM improves CAoI in unseen scenarios, unseen number of lanes, unseen number of vehicles, and joint dynamics up to 22.7\%, 26.4\%, 26.8\%, and 30.8\%, respectively.
These results demonstrate that the logical thinking ability of System~2 is critical for carrying knowledge and physics across different environments and network conditions rather than repetitive, statistical pattern learning over wireless data. The System~2 component of the proposed dual-mind approach encodes global structural, logical relations of the wireless networks, e.g., the temporal and spatial dependence of link scheduling, that remain valid across different environments and network topologies. This enables the reuse of symbolic abstractions and far fewer online learning steps. In a nutshell, with pretraining and few-shot adaptation, the proposed DMWM achieves stronger learning efficiency and generalization, while run-time execution remains the quick inference of System~1 with low latency. Hence, the proposed approach is practically applicable to wireless networks that require planning over extended horizons and adapt quickly to highly dynamic, complex environments.

\vspace{-0.4cm}
\section{Conclusion}
\vspace{-0.2cm}
In this paper, we have proposed a novel, unified world model-based learning approach for wireless networks, which overcomes the limitations of traditional RL approaches in data efficiency, long-term planning and generalization ability. Inspired by cognitive psychology, DMWM is composed of an intuitive, pattern-driven System~1 component and a logic-driven System~2 component. Taking the highly dynamic mmWave V2X network as an example, the proposed DMWM captures the dynamics and logical rules of the wireless network. Then, long-term link scheduling is learned in imagined trajectories with logical consistency over extended horizons rather than relying on expensive, real-time interactions with actual environments. Moreover, we have used the world model's imagination capability to jointly predict and schedule links when real-time wireless data is unavailable. Extensive simulation results on the realistic simulator show the significant improvements of DMWM in data efficiency and CAoI performance compared to the state-of-the-art RL baselines and the world model with only System~1. Moreover, the simulation results show the superior generalization and adaptivity of DMWM in unseen scenarios and conditions. Hence, the proposed DMWM-based learning approach has provided a promising new paradigm towards AGI-enabled wireless networks with complex dynamics and long-term optimization requirements.

\section*{Appendix A}
\vspace{-0.15cm}
\section*{Proof of Theorem 1}
\vspace{-0.15cm}
The conditional log-likelihood of the observed wireless features $o_{1:T}$ given the action sequence $a_{1:T}$ and the logic-enhanced generative model $\tilde p_\varphi$ is represented by
\begin{equation}
\ln \tilde p_\varphi\bigl(o_{1:T} \! \mid \! a_{1:T}\bigr) \! = \! \ln \int \tilde p_\varphi\bigl(o_{1:T},z_{1:T}\! \mid \! a_{1:T}\bigr)\,\mathrm{d}z_{1:T}.
\end{equation}
Then, we introduce the variational posterior
$
  q_\varphi\bigl(z_{1:T}\mid o_{1:T},a_{1:T}\bigr) = \prod_{t=1}^T q_\varphi\bigl(z_t\mid h_t,o_t\bigr), 
$
and we can rewrite the integral as an expectation in (19) as
\begin{equation}
\ln \tilde p_\varphi(o_{1:T}\mid a_{1:T}) = \ln \mathbb{E}_{q_\varphi(z_{1:T}\mid o_{1:T},a_{1:T})}
\Bigl[\tfrac{\tilde p_\varphi(o_{1:T},z_{1:T}\mid a_{1:T})}
{q_\varphi(z_{1:T}\mid o_{1:T},a_{1:T})}\Bigr].
\end{equation}
By applying Jensen's inequality, we obtain the ELBO as
\begin{equation}
\ln \tilde p_\varphi(o_{1:T}\mid a_{1:T})
\ge
\mathbb{E}_{q_\varphi}\Bigl[
\ln \tfrac{\tilde p_\varphi(o_{1:T},z_{1:T}\mid a_{1:T})}
{q_\varphi(z_{1:T}\mid o_{1:T},a_{1:T})}
\Bigr].
\end{equation}
We substitute the logic-enhanced model's factorization as
\begin{equation}
\tilde p_\varphi(o_{1:T},z_{1:T}\mid a_{1:T})=\prod_{t=1}^T p_\varphi\bigl(o_t\mid z_t\bigr)
p_\varphi\bigl(z_t\mid z_{t-1},a_{t-1}\bigr)
\phi^\alpha_t,
\end{equation}
and substitute the variational decomposition as
$q_\varphi(z_{1:T}\mid o_{1:T},a_{1:T}) =\prod_{t=1}^T q_\varphi\bigl(z_t\mid o_{\le t},a_{\le t}\bigr)$, then we obtain (23),
where we abbreviate $q_1=q_\varphi\bigl(z_t\mid o_{\le t},a_{<t}\bigr)$ and $q_2=q_\varphi\bigl(z_{t-1}\mid o_{\le t-1},a_{<t-1}\bigr)$.
\begin{figure*}
\begin{equation}
  {\setlength{\jot}{1pt}
\begin{aligned}
&\ln \tilde{p}_\varphi\left(o_{1: T} \mid a_{1: T}\right) \ge \mathbb{E}_{q_\varphi(z_{1:T} \mid o_{1:T}, a_{1:T})} \left[
\sum_{t=1}^T \ln p_\varphi(o_t  \mid  z_t) + \ln p_\varphi(z_t  \mid  z_{t-1}, a_{t-1}) + \ln \phi^{\alpha}_t- \ln q_\varphi(z_t  \mid  o_{\leq t}, a_{\leq t})
\right] \\
& = \sum_{t=1}^T(\mathbb{E}_{q_1}\left[\ln p_\varphi\left(o_t \mid z_t\right)\right]+ \mathbb{E}_{q_1}\left[\ln \phi^{\alpha}_{t} \right] -\mathbb{E}_{q_2}\left[ \mathrm{D}_{\textrm{KL}}\left[q_\varphi\left(z_t \mid o_{\leq t}, a_{<t}\right) \| p_\varphi\left(z_t \mid z_{t-1}, a_{t-1}\right)\right]\right])\\
\end{aligned}
  }
\end{equation}
\vspace{-0.1cm}
\hrulefill
\vspace{-0.6cm}
\end{figure*}

\bibliography{IEEEabrv,reference}

\end{document}